\documentclass[a4paper,UKenglish]{llncs}

\usepackage[T1]{fontenc}
\usepackage{etoolbox}

\usepackage{cancel} 

\usepackage{xifthen}
\usepackage{xparse}

\usepackage{amsmath}
\usepackage{amssymb}
\usepackage{mathtools}
\usepackage{proof} 
\usepackage{stmaryrd} 
\usepackage{upgreek} 
\usepackage{mathrsfs}

\usepackage[usenames,dvipsnames,svgnames,table]{xcolor}
\usepackage{tikz}
\usepackage[plain]{fancyref}
\usepackage{float}

\usepackage{fancyhdr}
\usepackage{hyperref} 
\usepackage{marginfix} 
\usepackage{microtype}
\usepackage{pdflscape} 
\usepackage{lastpage}

\usepackage{longtable}
\usepackage{tabularx}
\usepackage{multirow}
\usepackage{booktabs}

\usepackage{multicol} 

\usepackage{listings}

\usepackage{algorithm}
\usepackage[noend]{algpseudocode}

\usetikzlibrary{%
arrows,%
calc,%
shapes.misc,%
shapes.arrows,%
chains,%
matrix,%
positioning,%
scopes,%
decorations.pathmorphing,%
decorations.pathreplacing,%
shadows,%
er,%
trees,%
shapes,%
automata,%
patterns%
}

\pgfdeclarelayer{edgelayer}
\pgfdeclarelayer{nodelayer}
\pgfsetlayers{edgelayer,nodelayer,main}

\tikzset{initial text={}}

\tikzstyle{none}=[inner sep=0pt]
\tikzstyle{gate}=[rectangle,draw=Black, inner sep=5pt,outer sep=1pt]
\tikzstyle{and}=[gate]
\tikzstyle{or}=[gate]
\tikzstyle{seq}=[gate]
\tikzstyle{neg}=[gate]
\tikzstyle{fail}=[gate]
\tikzstyle{cost}=[gate, minimum height=2.5em, rectangle split, rectangle split parts=2, rectangle split draw splits=false, inner sep=3pt]
\tikzstyle{trigger}=[gate]
\tikzstyle{costguard}=[above,font=\scriptsize]
\tikzstyle{smallgate}=[gate,inner sep=4pt]
\tikzstyle{urgent}=[font={\scriptsize{\sf \textcolor{DimGray}{U}}}, inner sep=0pt]
\tikzstyle{state-circle}=[circle,draw=Black,inner sep=2pt,minimum size=0.35cm,outer sep=2pt]
\tikzstyle{state}=[ellipse,draw=Black,inner sep=2pt,outer sep=2pt]
\tikzstyle{state-box}=[rectangle,draw=Black,inner sep=2pt,minimum size=0.35cm,outer sep=2pt]
\tikzstyle{probability-state}=[circle,fill=Black,draw=Black,inner sep=0pt,minimum size=1pt]
\tikzstyle{BETrig}=[regular polygon,regular polygon sides=3,draw=Black, inner sep=0pt, outer sep=1pt,minimum width=0.85cm]
\tikzstyle{reset}=[rectangle,draw=Black, inner sep=5pt]
\tikzstyle{nodename}=[text width=1.8cm, align=center]
\tikzstyle{bename}=[nodename, below=0.3]
\tikzstyle{gatename}=[nodename, above=0.3]
\tikzstyle{costname}=[nodename, above=0.4]
\tikzstyle{update}=[font=\scriptsize]
\tikzstyle{false}=[color={Red!70!Black}]
\tikzstyle{true}=[color={Green!70!Black}]
\tikzstyle{undefined}=[color={Black!70}]
\tikzstyle{tt}=[defender]
\tikzstyle{ff}=[attacker]
\tikzstyle{uu}=[neutral]
\tikzstyle{attacker}=[fill={Red!30},postaction={pattern=horizontal lines,pattern color=Red!60}]
\tikzstyle{defender}=[fill={Green!30},postaction={pattern=vertical lines,pattern color=Green!60}]
\tikzstyle{neutral}=[fill={Gray!30},postaction={pattern=crosshatch dots,pattern color=Gray!80},font=\bf]
\tikzstyle{nothing}=[]
\tikzstyle{arrow-urgent}=[-latex,draw=Black]
\tikzstyle{arrow-non-urgent}=[-latex,draw=Black,dashed]
\tikzstyle{arrow-simple}=[-latex,draw=Black]
\tikzstyle{arrow-reset}=[-latex,dashed]
\tikzstyle{arrow-trigger}=[-latex,decoration={snake,segment length=4,amplitude=.6, post=lineto,post length=2pt},decorate]
\tikzstyle{arrow-segment-after-probabilities}=[-latex,draw=Black]
\tikzstyle{markovian}=[->,draw=Black,postaction={decorate},decoration={markings,mark=at position .5 with {\arrow{>}}}]
\tikzstyle{arrow-segment-before-probabilities-immediate}=[-,draw=Black]
\tikzstyle{arrow-segment-before-probabilities-timed}=[-,draw=Black]
\tikzstyle{stateText}=[rectangle,fill=White,draw=Black, inner sep=5pt, minimum height=0pt]
\tikzstyle{non-urgent}=[densely dashed]

\newcommand{\highlight}[1]{\textsf{\textbf{#1}}}

\newcommand*{\fancyrefapplabelprefix}{app}

\fancyrefaddcaptions{english}{%
  \providecommand*{\frefappname}{appendix}%
  \providecommand*{\Frefappname}{Appendix}%
}

\frefformat{plain}{\fancyrefapplabelprefix}{\frefappname\fancyrefdefaultspacing#1}
\Frefformat{plain}{\fancyrefapplabelprefix}{\Frefappname\fancyrefdefaultspacing#1}

\frefformat{vario}{\fancyrefapplabelprefix}{%
  \frefappname\fancyrefdefaultspacing#1#3%
}
\Frefformat{vario}{\fancyrefapplabelprefix}{%
  \Frefappname\fancyrefdefaultspacing#1#3%
}

\newcommand*{\fancyrefexlabelprefix}{ex}

\fancyrefaddcaptions{english}{%
  \providecommand*{\frefexname}{example}%
  \providecommand*{\Frefexname}{Example}%
}

\frefformat{plain}{\fancyrefexlabelprefix}{\frefexname\fancyrefdefaultspacing#1}
\Frefformat{plain}{\fancyrefexlabelprefix}{\Frefexname\fancyrefdefaultspacing#1}

\frefformat{vario}{\fancyrefexlabelprefix}{%
  \frefexname\fancyrefdefaultspacing#1#3%
}
\Frefformat{vario}{\fancyrefexlabelprefix}{%
  \Frefexname\fancyrefdefaultspacing#1#3%
}

\newcommand*{\fancyrefdeflabelprefix}{def}

\fancyrefaddcaptions{english}{%
  \providecommand*{\frefdefname}{definition}%
  \providecommand*{\Frefdefname}{Definition}%
}

\frefformat{plain}{\fancyrefdeflabelprefix}{\frefdefname\fancyrefdefaultspacing#1}
\Frefformat{plain}{\fancyrefdeflabelprefix}{\Frefdefname\fancyrefdefaultspacing#1}

\frefformat{vario}{\fancyrefdeflabelprefix}{%
  \frefdefname\fancyrefdefaultspacing#1#3%
}
\Frefformat{vario}{\fancyrefdeflabelprefix}{%
  \Frefdefname\fancyrefdefaultspacing#1#3%
}

\newcommand*{\fancyreftheolabelprefix}{theo}

\fancyrefaddcaptions{english}{%
  \providecommand*{\freftheoname}{theorem}%
  \providecommand*{\Freftheoname}{Theorem}%
}

\frefformat{plain}{\fancyreftheolabelprefix}{\freftheoname\fancyrefdefaultspacing#1}
\Frefformat{plain}{\fancyreftheolabelprefix}{\Freftheoname\fancyrefdefaultspacing#1}

\frefformat{vario}{\fancyreftheolabelprefix}{%
  \freftheoname\fancyrefdefaultspacing#1#3%
}
\Frefformat{vario}{\fancyreftheolabelprefix}{%
  \Freftheoname\fancyrefdefaultspacing#1#3%
}

\newcommand*{\fancyrefalglabelprefix}{alg}

\fancyrefaddcaptions{english}{%
  \providecommand*{\frefalgname}{algorithm}%
  \providecommand*{\Frefalgname}{Algorithm}%
}

\frefformat{plain}{\fancyrefalglabelprefix}{\frefalgname\fancyrefdefaultspacing#1}
\Frefformat{plain}{\fancyrefalglabelprefix}{\Frefalgname\fancyrefdefaultspacing#1}

\frefformat{vario}{\fancyrefalglabelprefix}{%
  \frefalgname\fancyrefdefaultspacing#1#3%
}
\Frefformat{vario}{\fancyrefalglabelprefix}{%
  \Frefalgname\fancyrefdefaultspacing#1#3%
}

\newcommand*{\fancyreflemlabelprefix}{lem}

\fancyrefaddcaptions{english}{%
  \providecommand*{\freflemname}{lemma}%
  \providecommand*{\Freflemname}{Lemma}%
}

\frefformat{plain}{\fancyreflemlabelprefix}{\freflemname\fancyrefdefaultspacing#1}
\Frefformat{plain}{\fancyreflemlabelprefix}{\Freflemname\fancyrefdefaultspacing#1}

\frefformat{vario}{\fancyreflemlabelprefix}{%
  \freflemname\fancyrefdefaultspacing#1#3%
}
\Frefformat{vario}{\fancyreflemlabelprefix}{%
  \Freflemname\fancyrefdefaultspacing#1#3%
}

\newcommand*{\fancyrefcorlabelprefix}{cor}

\fancyrefaddcaptions{english}{%
  \providecommand*{\frefcorname}{corollary}%
  \providecommand*{\Frefcorname}{Corollary}%
}

\frefformat{plain}{\fancyrefcorlabelprefix}{\frefcorname\fancyrefdefaultspacing#1}
\Frefformat{plain}{\fancyrefcorlabelprefix}{\Frefcorname\fancyrefdefaultspacing#1}

\frefformat{vario}{\fancyrefcorlabelprefix}{%
  \frefcorname\fancyrefdefaultspacing#1#3%
}
\Frefformat{vario}{\fancyrefcorlabelprefix}{%
  \Frefcorname\fancyrefdefaultspacing#1#3%
}

\renewcommand{\phi}{\upvarphi}
\renewcommand{\psi}{\uppsi}

\renewcommand{\rho}{\upvarrho}
\renewcommand{\upsilon}{\upupsilon}
\renewcommand{\chi}{\upchi}

\renewcommand{\max}{\mathop{\mathsf{max}}}
\renewcommand{\min}{\mathop{\mathsf{min}}}
\renewcommand{\sup}{\mathop{\mathsf{sup}}}
\renewcommand{\inf}{\mathop{\mathsf{inf}}}

\newcommand{\prism}{\textsf{PRISM}}
\newcommand{\VI}{\textsf{VI}}
\newcommand{\SI}{\textsf{SI}}
\newcommand{\CSG}{\textsf{CG}}
\newcommand{\SSG}{\textsf{SSG}}
\newcommand{\MDP}{\textsf{MDP}}
\newcommand{\EC}{\textsf{EC}}
\newcommand{\MEC}{\textsf{MEC}}

\newenvironment{lemma*}[1]{\medskip\noindent\textbf{Recap of Lemma\ifthenelse{\isempty{#1}}{}{~#1}.}}{}
\newenvironment{theorem*}[1]{\medskip\noindent\textbf{Recap of Theorem\ifthenelse{\isempty{#1}}{}{~#1}.}}{}

\DeclareDocumentCommand{\variable}{O{} D<>{}}{\mathsf{v}_{#1}^{#2}}
\DeclareDocumentCommand{\variables}{O{} D<>{}}{\mathcal{V}_{#1}^{#2}}

\newcommand{\set}[1]{\{#1\}}
\newcommand{\powerset}[1]{2^{#1}}
\newcommand{\disjointUnion}{\mathop{\dot{\cup}}}

\newcommand{\NN}{\mathbb{N}}

\DeclareDocumentCommand{\anySet}{O{} D<>{}}{\mathcal{X}_{#1}^{#2}}

\newcommand{\undefined}{\mathsf{undef}}

\DeclareDocumentCommand{\bijection}{O{} D<>{} D(){}}{\mathsf{bi}_{#1}^{#2}\ifthenelse{\isempty{#3}}{}{(#3)}}

\DeclareMathOperator*{\argmax}{arg\,max}
\DeclareMathOperator*{\argmin}{arg\,min}

\newcommand{\pspace}{\textsl{PSPACE}}

\DeclareDocumentCommand{\alg}{O{} D<>{} D(){}}{\mathcal{A}_{#1}^{#2}\ifthenelse{\isempty{#3}}{}{(#3)}}

\DeclareDocumentCommand{\event}{O{} D<>{} D(){}}{\mathcal{E}_{#1}^{#2}\ifthenelse{\isempty{#3}}{}{(#3)}}
\DeclareDocumentCommand{\distribution}{O{} D<>{} D(){}}{\mu_{#1}^{#2}\ifthenelse{\isempty{#3}}{}{(#3)}}
\DeclareDocumentCommand{\probability}{O{} D<>{} D(){} D(){}}{\mathbb{P}_{#1}^{#2}\ifthenelse{\isempty{#3}}{}{\big(#3\big)}\ifthenelse{\isempty{#4}}{}{\big(#4\big)}}
\DeclareDocumentCommand{\preop}{O{} D<>{} D(){} D(){}}{\mathsf{Pre}_{#1}^{#2}\ifthenelse{\isempty{#3}}{}{(#3)}\ifthenelse{\isempty{#4}}{}{(#4)}}
\DeclareDocumentCommand{\expectation}{O{} D<>{}  D(){}}{\mathbb{E}_{#1}^{#2}\ifthenelse{\isempty{#3}}{}{[#3]}}
\DeclareDocumentCommand{\support}{O{} D<>{} D(){}}{\mathsf{Supp}_{#1}^{#2}\ifthenelse{\isempty{#3}}{}{(#3)}}
\newcommand{\dirac}[1]{\delta_{#1}}

\DeclareDocumentCommand{\cylinder}{O{} D<>{} D(){}}{\mathsf{Cyl}_{#1}^{#2}\ifthenelse{\isempty{#3}}{}{\big(#3\big)}}
\DeclareDocumentCommand{\algebra}{O{} D<>{} D(){}}{\mathcal{F}_{#1}^{#2}\ifthenelse{\isempty{#3}}{}{\big(#3\big)}}
\DeclareDocumentCommand{\probspace}{O{} D<>{}}{\big(\Omega_{#1}^{#2},\mathcal{F}_{#1}^{#2},\probability\big)}
\DeclareDocumentCommand{\density}{O{} D<>{} D(){}}{\mathsf{f}_{#1}^{#2}\ifthenelse{\isempty{#3}}{}{(#3)}}

\DeclareDocumentCommand{\randomVariable}{O{} D<>{} D(){} D(){}}{\mathsf{C}_{#1}^{#2}\ifthenelse{\isempty{#3}}{}{(#3)}\ifthenelse{\isempty{#4}}{}{(#4)}}

\DeclareDocumentCommand{\valuation}{t* O{} D<>{} D(){}}{\IfBooleanTF{#1}{\eta}{\upsilon}_{#2}^{#3}\ifthenelse{\isempty{#4}}{}{(#4)}}
\DeclareDocumentCommand{\valuations}{O{} D<>{} D(){}}{\mathsf{Val}_{#1}^{#2}\ifthenelse{\isempty{#3}}{}{(#3)}}

\DeclareDocumentCommand{\player}{O{} D<>{} D(){} t* O{}}{\mathsf{p}_{#1}^{#2}\ifthenelse{\isempty{#3}}{}{(#3)}\ifthenelse{\isempty{#5}}{}{[#5]}}
\DeclareDocumentCommand{\execution}{O{} D<>{} D(){} t* O{}}{\pi_{#1}^{#2}\ifthenelse{\isempty{#3}}{}{(#3)}\ifthenelse{\isempty{#5}}{}{[#5]}}
\DeclareDocumentCommand{\executions}{O{} D<>{} D(){\game}}{\mathsf{Play}_{#1}^{#2}\ifthenelse{\isempty{#3}}{}{(#3)}}

\DeclareDocumentCommand{\exampleExecution}{O{} D<>{} D(){} D||{n} t*}{
 \state[0]\IfBooleanTF{#5}{\state[1]\state[2]\cdots}{\cdots\state[#4]}}

\DeclareDocumentCommand{\last}{O{} D<>{} D(){}}{\mathsf{last}_{#1}^{#2}\ifthenelse{\isempty{#3}}{}{(#3)}}

\DeclareDocumentCommand{\path}{O{} D<>{} D(){}}{\pi_{#1}^{#2}\ifthenelse{\isempty{#3}}{}{(#3)}}
\DeclareDocumentCommand{\paths}{O{} D<>{} D(){\game}}{\mathsf{Paths}_{#1}^{#2}\ifthenelse{\isempty{#3}}{}{(#3)}}
\DeclareDocumentCommand{\examplePath}{O{} D<>{} D(){} D||{n}}{\path[#1]<#2>(#3) = \state[0]\dots\state[#4]}

\DeclareDocumentCommand{\position}{O{} D<>{} D(){}}{\mathcal{X}_{#1}^{#2}\ifthenelse{\isempty{#3}}{}{(#3)}}

\DeclareDocumentCommand{\state}{O{} D<>{} D(){}}{\mathsf{s}_{#1}^{#2}\ifthenelse{\isempty{#3}}{}{(#3)}}
\DeclareDocumentCommand{\states}{O{} D<>{} D(){}}{\mathsf{S}_{#1}^{#2}\ifthenelse{\isempty{#3}}{}{(#3)}}
\DeclareDocumentCommand{\initialState}{O{} D<>{} D(){}}{\mathsf{s}_{0\ifthenelse{\isempty{#1}}{}{,#1}}^{#2}\ifthenelse{\isempty{#3}}{}{(#3)}}

\DeclareDocumentCommand{\move}{O{} D<>{} D(){}}{\mathsf{m}_{#1}^{#2}\ifthenelse{\isempty{#3}}{}{(#3)}}
\DeclareDocumentCommand{\moves}{O{} D<>{} D(){}}{\mathsf{M}_{#1}^{#2}\ifthenelse{\isempty{#3}}{}{(#3)}}

\DeclareDocumentCommand{\agent}{O{} D<>{} D(){}}{\mathsf{p}_{#1}^{#2}\ifthenelse{\isempty{#3}}{}{(#3)}}
\DeclareDocumentCommand{\agents}{O{} D<>{} D(){}}{\set{\attacker,\defender}_{#1}^{#2}\ifthenelse{\isempty{#3}}{}{(#3)}}

\DeclareDocumentCommand{\transitions}{O{} D<>{} D(){} D(){} t*}{\mathsf{T}_{#1}^{#2}\ifthenelse{\isempty{#3}}{}{\big(#3\big)}\ifthenelse{\isempty{#4}}{}{\big(#4\big)}}

\DeclareDocumentCommand{\play}{O{} D<>{} D(){}}{\rho_{#1}^{#2}\ifthenelse{\isempty{#3}}{}{(#3)}}
\DeclareDocumentCommand{\plays}{O{} D<>{} D(){}}{\mathsf{Plays}_{#1}^{#2}\ifthenelse{\isempty{#3}}{}{(#3)}}

\DeclareDocumentCommand{\strategy}{O{} D<>{} D(){} D(){}
  t*}{\IfBooleanTF{#5}{\sigma}{\IfNoValueTF{#1}{\rho}{\ifstrequal{#1}{\safe}{\sigma}{\ifstrequal{#1}{\maximize}{\pi}{\ifstrequal{#1}{\minimize}{\pi}{\rho}}}}}_{\IfNoValueTF{#1}{}{\ifstrequal{#1}{\safe}{}{\ifstrequal{#1}{\reach}{}{#1}}}}^{#2}\ifthenelse{\isempty{#3}}{}{(#3)}\ifthenelse{\isempty{#4}}{}{(#4)}}
\DeclareDocumentCommand{\strategies}{O{} D<>{} D(){} D(){}
  t*}{\IfBooleanTF{#5}{\mathcal{S}}{\IfNoValueTF{#1}{\mathcal{R}}{\ifstrequal{#1}{\safe}{\mathcal{S}}{\ifstrequal{#1}{\maximize}{\Pi}{\ifstrequal{#1}{\minimize}{\Pi}{\mathcal{R}}}}}}_{\IfNoValueTF{#1}{}{\ifstrequal{#1}{\safe}{}{\ifstrequal{#1}{\reach}{}{#1}}}}^{#2}\ifthenelse{\isempty{#3}}{}{(#3)}\ifthenelse{\isempty{#4}}{}{(#4)}}

\DeclareDocumentCommand{\selector}{O{} D<>{} D(){}}{\chi_{#1}^{#2}\ifthenelse{\isempty{#3}}{}{(#3)}}
\DeclareDocumentCommand{\selectors}{O{} D<>{} D(){}}{\mathcal{X}_{#1}^{#2}\ifthenelse{\isempty{#3}}{}{(#3)}}

\DeclareDocumentCommand{\game}{O{} D<>{} D(){}}{\mathsf{G}_{#1}^{#2}\ifthenelse{\isempty{#3}}{}{(#3)}}
\DeclareDocumentCommand{\games}{O{} D<>{} D(){}}{\mathcal{G}_{#1}^{#2}\ifthenelse{\isempty{#3}}{}{(#3)}}
\DeclareDocumentCommand{\SG}{O{} D<>{} D(){}}{\mathsf{SG}_{#1}^{#2}\ifthenelse{\isempty{#3}}{}{(#3)}}

\DeclareDocumentCommand{\transitionProbability}{O{} D<>{}  D(){}}{\mathsf{P}_{#1}^{#2}\ifthenelse{\isempty{#3}}{}{(#3)}}

\DeclareDocumentCommand{\safe}{O{} D<>{} D(){}}{\mathscr{S}_{#1}^{#2}\ifthenelse{\isempty{#3}}{}{[#3]}}
\DeclareDocumentCommand{\reach}{O{} D<>{}  D(){}}{\mathscr{R}_{#1}^{#2}\ifthenelse{\isempty{#3}}{}{[#3]}}
\DeclareDocumentCommand{\fail}{O{} D<>{} D(){}}{\mathsf{Safe}_{#1}^{#2}\ifthenelse{\isempty{#3}}{}{[#3]}}
\DeclareDocumentCommand{\success}{O{} D<>{}  D(){}}{\mathsf{Reach}_{#1}^{#2}\ifthenelse{\isempty{#3}}{}{[#3]}}
\DeclareDocumentCommand{\val}{O{} D<>{} D(){} D(){}
  t*}{\mathsf{val}_{#1}^{#2}\ifthenelse{\isempty{#3}}{}{(#3)}\ifthenelse{\isempty{#4}}{}{(#4)}}
\DeclareDocumentCommand{\attractor}{O{} D<>{} D(){} D(){}
  t*}{\mathsf{Attr}_{#1}^{#2}\ifthenelse{\isempty{#3}}{}{(#3)}\ifthenelse{\isempty{#4}}{}{(#4)}}
\DeclareDocumentCommand{\rank}{O{} D<>{} D(){} D(){}
  t*}{\mathsf{rank}_{#1}^{#2}\ifthenelse{\isempty{#3}}{}{(#3)}\ifthenelse{\isempty{#4}}{}{(#4)}}
\DeclareDocumentCommand{\maximize}{O{} D<>{} D(){}}{\top_{#1}^{#2}\ifthenelse{\isempty{#3}}{}{[#3]}}
\DeclareDocumentCommand{\minimize}{O{} D<>{} D(){}}{\bot_{#1}^{#2}\ifthenelse{\isempty{#3}}{}{[#3]}}
\DeclareDocumentCommand{\winning}{O{\safe} D<>{} D(){}}{\mathsf{W}_{#1}^{#2}\ifthenelse{\isempty{#3}}{}{[#3]}}
\DeclareDocumentCommand{\differ}{O{} D<>{} D(){}}{\Delta_{#1}^{#2}\ifthenelse{\isempty{#3}}{}{(#3)}}

\DeclareDocumentCommand{\moveAssignment}{O{} D<>{} D(){}}{\Gamma_{#1}^{#2}\ifthenelse{\isempty{#3}}{}{(#3)}}
\DeclareDocumentCommand{\exampleGame}{O{} D<>{} D(){} t*}{\game[#1]<#2>(#3)=(\states[#1]<#2>\IfBooleanTF{#4}{(#3)}{},\moves[#1]<#2>\IfBooleanTF{#4}{(#3)}{},\moveAssignment[\reach\ifthenelse{\isempty{#1}}{}{,#1}]<#2>(#3),\moveAssignment[\safe\ifthenelse{\isempty{#1}}{}{,#1}]<#2>(#3),\transitions[#1]<#2>\IfBooleanTF{#4}{(#3)}{})}
\DeclareDocumentCommand{\destination}{O{} D<>{} D(){}}{\mathsf{Post}_{#1}^{#2}\ifthenelse{\isempty{#3}}{}{(#3)}}
\DeclareDocumentCommand{\endComponent}{O{} D<>{} D(){}}{\mathsf{C}_{#1}^{#2}\ifthenelse{\isempty{#3}}{}{(#3)}}
\DeclareDocumentCommand{\bestExit}{O{} D<>{} D(){}}{\mathsf{bestExit}_{\ifstrequal{#1}{\reach}{}{#1}}^{#2}\ifthenelse{\isempty{#3}}{}{(#3)}}
\DeclareDocumentCommand{\upperBound}{O{} D<>{} D(){}}{\mathsf{U}_{#1}^{#2}\ifthenelse{\isempty{#3}}{}{(#3)}}
\DeclareDocumentCommand{\lowerBound}{O{} D<>{} D(){}}{\mathsf{L}_{#1}^{#2}\ifthenelse{\isempty{#3}}{}{(#3)}}
\DeclareDocumentCommand{\exit}{O{} D<>{} D(){}}{\mathsf{exit}\ifthenelse{\isempty{#2}}{}{(#2)}\ifthenelse{\isempty{#3}}{}{(#3)}}
\DeclareDocumentCommand{\best}{O{} D<>{} D(){}}{\mathsf{best}_{#1}^{#2}\ifthenelse{\isempty{#3}}{}{(#3)}}
\DeclareDocumentCommand{\reasonable}{O{} D<>{} D(){}}{\mathsf{reasonable}_{#1}^{#2}\ifthenelse{\isempty{#3}}{}{(#3)}}
\DeclareDocumentCommand{\upperSet}{O{} D<>{} D(){}}{\mathfrak{U}_{#1}^{#2}\ifthenelse{\isempty{#3}}{}{(#3)}}
\DeclareDocumentCommand{\lowerSet}{O{} D<>{} D(){}}{\mathfrak{L}_{#1}^{#2}\ifthenelse{\isempty{#3}}{}{(#3)}}
\DeclareDocumentCommand{\onestep}{m D<>{} D(){}}{\mathsf{#1}^{#2}\ifthenelse{\isempty{#3}}{}{(#3)}}

\DeclareDocumentCommand{\evaluate}{O{} D<>{} D(){\valuation[\old]} D(){\valuation[\basicEvents]}
  D(){}}{\mathsf{app}_{#1}^{#2}\ifthenelse{\isempty{#3}}{}{(#3,#4)}\ifthenelse{\isempty{#5}}{}{(#5)}}

\newcommand{\old}{\mathsf{o}}

\DeclareDocumentCommand{\side}{O{} D<>{} m m D(){}}{\mathsf{side}_{#1}^{#2}\ifthenelse{\isempty{#3}}{}{(#3,#4)}\ifthenelse{\isempty{#5}}{}{(#5)}}

\newcolumntype{L}[1]{>{\hsize=#1\hsize\raggedright\arraybackslash}X}%
\newcolumntype{R}[1]{>{\hsize=#1\hsize\raggedleft\arraybackslash}X}%
\newcolumntype{C}[1]{>{\hsize=#1\hsize\centering\arraybackslash}X}%

\algblockdefx[Structure]{Structure}{EndStructure}[1][Unknown]{\textsf{\textbf{Structure
      #1 \{}}}{\textsf{\textbf{\}}}}
\algblockdefx[Case]{Case}{EndCase}[1][Unknown]{\textsf{\textbf{case}}
      #1 \textsf{\textbf{of}}}{}

\algblockdefx[CaseOf]{CaseOf}{EndCaseOf}[1][Otherwise]{#1:}{}

\algblockdefx[DoUntil]{DoUntil}{EndDoUntil}{\highlight{repeat}}{\highlight{until\ }}
\newcommand{\procname}[1]{\textnormal{\textsf{#1}}}

\title{Stopping Criteria for Value and Strategy Iteration on
  Concurrent Stochastic Reachability Games} \author{Julia Eisentraut
  and Jan K\v ret\'insk\'y and Alexej Rotar} \institute{Technical
  University of Munich}

\begin{document}
\maketitle

\pagestyle{plain}

\renewcommand{\argmax}{\mathop{\mathsf{arg \, max}}}
\renewcommand{\argmin}{\mathop{\mathsf{arg \, min}}}

\begin{abstract}
  We consider concurrent stochastic games played on graphs with
  reachability and safety objectives. These games can be solved by
  \emph{value iteration} as well as \emph{strategy iteration}, each of
  them yielding a sequence of under-approximations of the reachability
  value and a sequence of over-approximation of the safety value,
  converging to it in the limit. For both approaches, we provide the
  first (anytime) algorithms with stopping criteria. The stopping
  criterion for value iteration is based on providing a convergent
  sequence of over-approximations, which then allows to estimate the
  distance to the true value. For strategy iteration, we bound the
  error by complementing the strategy iteration algorithm for
  reachability by a new strategy iteration algorithm
  under-approximating the safety-value.
\end{abstract}


\section{Introduction}
\label{sect:intro}
A \emph{concurrent stochastic game} \cite{AlfaroHK98} is a two-player
game played on a graph.  At every round of the game, each player
simultaneously and independently chooses a move.  The moves then
jointly determine the transition taken, which leads to a probability
distribution over states.  We consider \emph{safety} and
\emph{reachability objectives} \cite{AlfaroHK98}.  Considering a
safety objective for player~$\safe$, its goal is to maximize the
probability of staying within a given set of states, while
player~$\reach$ maximizes the probability to leave this set, which is
its reachability objective.  Hence, the two objectives are dual, the
games are symmetric by swapping the players and thus, from now on we
refer to both simply as concurrent games (\CSG{}).  These games are
determined~\cite{Everett1957}, i.e. the supremum probability which
player~$\safe$ can ensure for staying in the safe set is equal to one
minus the supremum probability which player~$\reach$ can ensure for
reaching a state outside.  Deciding
whether 
this value is at least~$p$ for $p \in [0,1]$ is in
\pspace~\cite{Etessami2006}.  For the reachability objective,
player~$\reach$ is only guaranteed the existence of $\epsilon$-optimal
(memoryless randomized) strategies~\cite{Everett1957}.  For
player~$\safe$, optimal (again memoryless randomized) strategies
exist~\cite{Parthasarathy1973}.

Algorithms for concurrent reachability games have been further studied
and their termination discussed
in~\cite{Chatterjee2006,Chatterjee2009,Chatterjee2013}. 
The algorithms for solving the games are based on dynamic programming, namely \emph{value iteration} (\VI{}) and \emph{strategy iteration} (\SI{}):

Firstly, \VI{} produces a non-decreasing sequence that
under-approximates the optimal probability to reach the given states
and in the limit converges to it \cite{Alfaro2004}.  However, no
stopping criterion is known for this process.  Hence, the current
error cannot be bounded at a given moment.  Although this sequence
yields by determinacy an over-approximation for the value of the
safety objective, there are no known sequences over-approximating the
reachability value or dually under-approximating the safety value that
would converge to the actual values.  Our first contribution is an
algorithm producing such a sequence, thus yielding \emph{the first
  stopping criterion for \VI{}} for these games and an \emph{anytime
  \VI{} algorithm}, which at any moment can bound the current
imprecision in the approximation, converging to 0.  Indeed, whenever
the under- and over-approximations are less than $\epsilon$ apart, for
$\epsilon > 0$, they are also $\epsilon$-close to the actual value of
the game.

Secondly, \SI{} produces a sequence of strategies guaranteeing
non-decreasing probabilities to reach the given states, converging in
the limit to the $\epsilon$-optimum.  \SI{} can thus provide
under-approximating sequence for reachability.  However, similarly to
\VI{}, the known approaches only work for reachability and not for
safety.  Our second contribution is an \SI{} algorithm, which
converges to the safety-value from below. Again, this yields \emph{a
  stopping criterion for \SI{}} and an \emph{anytime \SI{} algorithm}.


\paragraph*{Our Approach}
As mentioned above, the over-approximations coming from known \VI{}
algorithms for reachability as well as the under-approximations coming
from known \SI{} algorithms for safety~\cite{Chatterjee2009} are not
converging to the true value of the game~\cite{Chatterjee2013}.  The
reason for this
is the presence of so-called \emph{end components
  (\EC{})}~\cite{alfarothesis}.  In technical terms, due to \EC{}s the
greatest fixpoint of the \VI{} operator (also called Bellman update)
is different from the least one. While the over-approximations
converge to the greatest fixpoint, the true value is the least
fixpoint. This problem actually exists even for the much simpler
single-player case of \emph{Markov decision processes
  (\MDP{})}~\cite{Puterman}.

For \MDP{}, this issue has been solved by collapsing each \EC{} into a
single state, effectively erasing indefinite cycles
\cite{Brazdil2014,Haddad2018}.  This prevents states of an \EC{} to
rely on each other's unsubstantiated overly high estimate of the
value, reduces their estimate at once to that of the actions
\emph{leaving} the EC, causing all the fixpoints to coincide.  This
has been observed insufficient \cite{Kelmendi2018} for \emph{simple
  stochastic games (\SSG{})} \cite{Condon90}, i.e. ``turn-based''
\CSG{} where in each state only one player has a non-trivial choice.
As opposed to \MDP{}, states of the same \EC{} in an \SSG{} may have
different values, hence, cannot be all collapsed and their estimates
reduced to the same value.  Instead, \cite{Kelmendi2018} proposes to
gradually \emph{deflate} (decrease) each estimate whenever it is not
substantiated by a move with that estimate \emph{leaving} the EC (with
positive probability).  Since there are different leaving
moves with different values, this gives rise to different parts of an
\EC{} called \emph{simple \EC{}}, each corresponding to a sphere of
influence of each leaving move with potentially different values.

%

In our setting, the main challenge is to find the analogue to the
simple \EC{} and how to deflate them in the right
way. In~\cite{Kelmendi2018}, \EC{}s are set of states such that there
exists a set of moves, which only lead to states inside the \EC{}, but
still, for any two states in the \EC{}, there exists a finite path
between them only taking transitions labeled with the given
moves. This definition of \EC{}s reveals already one big obstacle when
it comes to \CSG{}: The set of states, which transitions with a given
move lead to, depends on the other player's simultaneous and
independent choice. Hence, the given definition of \EC{} does not
prove to be suitable in our setting. Instead, a matrix game has to be
solved repetitively for each state to determine the best distribution
over available moves. In this matrix game, we have to face another
issue, namely that an extremum over all strategies leaving an \EC{}
with a positive (arbitrary small) probability may be realized only by
an optimal strategy that is leaving with zero probability, i.e. not
leaving at all. For instance, consider \Fref{fig:problemECs}, for
$\epsilon \to 0$, the strategy, which assigns $\epsilon$ to move~$b$
and $(1-\epsilon)$ to move~$a$ yields an increasingly better value for
the matrix game at $\state[1]$ with respect to strategies assigning a
positive probability to states outside the \EC{}. The supremum of this
sequence is~$1$, however the strategy achieving it is not exiting
anymore. Once the matrix game is solved, the sphere of influence of a
leaving convex combination of moves can be computed by the classical
attractor construction, yielding the desired analogue of the simple
\EC{}. This can then be finally deflated according to our notion of
the best value when leaving the \EC{}. Finally, to the best of our
knowledge for the first time, we adapt such a parallel
under-/over-approximating \VI{} approach to \SI{} on \CSG{}.

Our contribution can be summarized as follows:
\begin{itemize}
\item We introduce a \VI{} algorithm yielding both under- and
  over-approximation sequences, both of which converge to the value of
  the game. Thus, we present the first stopping criterion for \VI{} on
  \CSG{} and the first anytime algorithm with guaranteed
  precision. 
\item We introduce an \SI{} algorithm for safety strategies in \CSG{}.
  Since these results in both under- and over-approximation sequences
  for both objectives, we analogously obtain the first stopping
  criterion for \SI{} on \CSG{} and the first anytime algorithm with
  guaranteed precision.
\item As direct consequences, we obtain (i) that for \CSG{} without
  non-trivial end components, the simpler solution (without
  deflating), is sufficient, and (ii) an \SI{} algorithm for safety
  \SSG{} that is simpler than Algorithm~2 in \cite{Chatterjee2013}, which needs to
  transform the game.
\end{itemize}

\paragraph*{Further Related Work}
The \pspace{}-algorithm given in~\cite{Etessami2006} to decide whether
the value of a given recursive game is at least $p$ for $p \in [0,1]$
allows for a trivial stopping criterion by iteratively executing this
algorithm for a suitable sequence of~$(p_i)_{i \in \NN}$ (intuitively,
we try to choose~$p_i$ such at alternatingly, the value of the game is
above and below the true value, while the distance between to
succeeding $p_i$ monotonically decreases). However, this criterion is
impractical since it definitely need exponential time. The following
stopping criteria we present allow for a potentially fast
approximation.

The idea of complementing the under-approximating sequence of \VI{} by
an over-approximating one dates back to \cite{McMahan2005} as
\emph{bounded \VI{}} (due to the new upper bound).  It does not
converge for general \MDP{}, but in fact only for \MDP{} without
\EC{}s as often considered in the stochastic shortest path problem.
The convergence is ensured in \cite{Brazdil2014,Haddad2018} by
collapsing \EC{}s, in \cite{Brazdil2014} on the fly, in
\cite{Haddad2018} as a preprocessing step, calling it interval
iteration. 

The first practical stopping criterion for \SI{} in \SSG{} ((but not
for \CSG{})) is given in \cite{Chatterjee2013}. To this end, an \SI{}
algorithm for safety strategies is given, which relies on a repetitive
transformation of the underlying game.
That the given algorithm does not work properly for concurrent
stochastic games has been observed in \cite{Chatterjee2013},
correcting the claims of \cite{Chatterjee2009}. Further, this approach
is claimed not extensible to \VI{}. The first \VI{} stopping criterion
in \SSG{} is obtained in \cite{Kelmendi2018}, which we extend here to
\CSG{}.

A generalization of \CSG{} to $\omega$-regular objectives has been
considered in \cite{Alfaro2000}. Value iteration via quantitative game
$\mu$-calculus has been discussed in~\cite{Alfaro2004}. As to tool
support, the only model checker for \CSG{} is
\prism{}-games~\cite{Kwiatkowska2018}. Model checking implementations
for MDP that take stopping criteria into account are extensions of
PRISM~\cite{Baier17} and Storm \cite{storm,QuatmannK18}.




\section{Stochastic Games}
\label{sec:preliminaries}

In this section, we recall basic notions related to stochastic games.
For a countable set $X$, a function $\mu \colon X \to [0,1]$ is called
a \emph{distribution} over~$X$ if $\sum_{x \in X} \mu(x) = 1$.  The
\emph{support} of $\mu$ is $\support(\mu) = \set{x \mid \mu(x)>0}$.
The set of all distributions over $X$ is denoted by
$\distribution(X)$. If there is a unique $x \in X$ such that
$\distribution(x) = 1$, we call the distribution \emph{Dirac} and
denote it by $\dirac{x}$.

\begin{definition}[(Two-Player Stochastic) Concurrent Game]
  A \emph{concurrent game} is a tuple $\exampleGame$, where $\states$
  is a finite set of \emph{states}, $\moves$ is a finite set of
  \emph{moves}, $\moveAssignment[\reach],\moveAssignment[\safe] \colon \states
  \to \powerset{\moves} \setminus \emptyset$ are two move assignments
  and $\transitions \colon \states \times \moves \times \moves \to
  \distribution(\states)$ is a transition function.    
  For $\player \in \set{\reach,\safe}$, assignment $\moveAssignment[\player]$ associates
  each state $\state \in \states$ with a nonempty set
  $\moveAssignment[\player](\state) \subseteq \moves$ of moves \emph{available} to
  player~$\player$ at state~$\state$.
  $\transitions(\state,\move[\reach],\move[\safe])(\state<\prime>)$ gives the
  \emph{probability of a transition} from state~$\state$ to
  state~$\state<\prime>$ when player~$\reach$ chooses move~$\move[\reach]\in
  \moveAssignment[\reach](\state)$ and player~$\safe$  move~\mbox{$\move[\safe]\in
  \moveAssignment[\safe](\state)$}.
\end{definition}
A concurrent game is \emph{turn-based} if for every state~$\state$
there exists $p \in \set{\reach,\safe}$ such that
$|\moveAssignment[p](\state)| = 1$; then we call it a \emph{turn-based
  game}, rather than a turn-based concurrent game.
A \emph{play}~$\execution$ of $\game$ is an infinite sequence
$\exampleExecution*$ of states such that for all $i \in \NN$ there are
moves $\move[\reach]<i> \in \moveAssignment[\reach](\state[i])$ and $\move[\safe]<i>
\in \moveAssignment[\safe](\state[i])$ with
$\transitions(\state[i],\move[\reach]<i>,\move[\safe]<i>)(\state[i+1]) > 0$. We
denote by $\executions$ the set of all plays
and by
$\executions[\state]$ the set of all plays $\exampleExecution*$
such that $\state[0]=\state$. 
A \emph{strategy} for player~$\player$ is a function
$\strategy[\player] \colon \states \to \distribution(\moves)$ that
assigns to each state a distribution over moves available to
player~$\player$,\footnote{Since memoryless strategies are sufficient
  for the objectives considered in this paper, we do not introduce
  general history-dependent strategies to avoid clutter.}  i.e. for
all $\state \in \states$, we have $\support(\strategy(\state))
\subseteq \moveAssignment[\player](\state)$. We call a strategy
\emph{pure} if all distributions it returns are Dirac.
In the following, we denote by $\strategies[\reach]$ the set of
strategies for player~$\reach$ and by $\strategies[\safe]$ the set of
strategies for player~$\safe$. In addition, we use $\strategy$ to
denote a single strategy of player~$\reach$ and $\strategy*$ to denote a
single strategy of player~$\safe$.


\paragraph*{Semantics.}
Given two strategies~$\strategy[\reach]$ and $\strategy[\safe]$ and a
starting state~$\initialState$, we give the concurrent game the standard 
semantics in terms of a Markov chain with the same state space $\states$, 
the initial state $\initialState$, and the transition probabilities $P$ given by 
\[
P(s,s')=\sum_{\move[\reach],\move[\safe] \in \moves} 
  \transitions(\state,\move[\reach],\move[\safe])(\state<\prime>) \cdot
  \strategy[\reach](\state)(\move[\reach]) \cdot \strategy[\safe](\state)(\move[\safe])
\] 
We denote by
$\probability[\initialState]<\strategy[\reach],\strategy[\safe]>$ the
standard probability measure over the plays induced by this Markov chain
and define this to be the probability measure over plays of the game 
when player~$\reach$ plays strategy~$\strategy[\reach]$, player~$\safe$ plays 
strategy~$\strategy[\safe]$ and the game starts in state~$\initialState$.

\paragraph*{Reachability and Safety Objectives.}
Let $\fail, \success \subseteq \states$ form a partitioning of $\states$. 
$\success$ denotes the set of
states player~$\reach$ wants to reach, while $\fail$ denotes the set
of states player~$\safe$ wants to confine the game in. 
We denote the \emph{reachability} objective by $\Diamond \success
\coloneqq \set{ \exampleExecution* \mid \exists i \in \NN: \state[i]
  \in \success}$ and the \emph{safety} objective by $\Box \fail
\coloneqq \set{\exampleExecution* \mid \forall i \in \NN: \state[i]
  \in \fail}$. The \emph{value} of the 
objective $\Diamond \success$ at state~$\state$ is given by
\[
    \val(\Diamond \success)(\state) \coloneqq \adjustlimits\sup_{\strategy[\reach] \in \strategies[\reach]} \inf_{\strategy[\safe]\in\strategies[\safe]} \probability[\state]<\strategy[\reach],\strategy[\safe]>(\Diamond \success)
\]

and the \emph{value} of 
the objective $\Box\fail$ by
\[
    \val(\Box \fail)(\state) \coloneqq \adjustlimits\sup_{\strategy[\safe] \in \strategies[\safe]} \inf_{\strategy[\reach]\in\strategies[\reach]} \probability[\state]<\strategy[\safe],\strategy[\reach]>(\Box \fail).
\]
Additionally, we define the value given a fixed strategy as
$\val[\reach:\strategy[\reach]](\Diamond \success)(\state) \coloneqq
\inf_{\strategy[\safe]\in\strategies[\safe]}
\probability[\state]<\strategy[\reach],\strategy[\safe]>(\Diamond
\success)$ and $\val[\safe:\strategy[\safe]](\Box \fail)(\state)
\coloneqq \inf_{\strategy[\reach]\in\strategies[\reach]}
\probability[\state]<\strategy[\safe],\strategy[\reach]>(\Box \fail)$.
By the determinacy of these games~\cite{Everett1957} and the duality
of these objectives, we have $\val[\reach](\Diamond \success)(\state)
+ \val[\safe](\Box \fail)(\state) = 1$ (since $\success$ and $\fail$
partition the state space).

Let $\state \in \states$, $\move[\reach] \in
\moveAssignment[\reach](\state)$ and $\move[\safe] \in
\moveAssignment[\safe](\state)$. We denote the \emph{set of potential
  successors} of~$\state$ by
$\destination(\state,\move[\reach],\move[\safe]) =
\support(\transitions(\state, \move[\reach], \move[\safe]))$. In
addition, we lift the notation to strategies~$\strategy[\reach]$ and
$\strategy[\safe]$ by \[ \destination(\state,\strategy[\reach],
\strategy[\safe]) = \bigcup_{\move[\reach] \in
  \support(\strategy[\reach](\state))} \bigcup_{\move[\safe] \in
  \support(\strategy[\safe](\state))}
\destination(\state,\move[\reach],\move[\safe]).\]

We denote by $\winning \coloneqq \set{\state \mid \state \in \states
  \land \val(\Diamond \success)(\state) = 0}$ the \emph{sure winning}
region of player~$\safe$. It can be computed in at most
$|\states|$-steps by iteration $\winning<0> \coloneqq \states
\setminus \success$ and $\winning<k+1> = \set{s \in \states \setminus
  \success \mid \exists \move[\safe] \in
  \moveAssignment[\safe](\state): \forall \move[\reach] \in
  \moveAssignment[\reach](\state):
  \destination(\state,\move[\reach],\move[\safe]) \subseteq
  \winning<k>}$ for all $k \in \NN$~\cite{Alfaro2000}. Consequently,
we can assume without loss of generality that $\success$ and
$\winning$ are both singletons and absorbing.

\paragraph*{End Components.}
Let $\exampleGame$ be a concurrent game. A non-empty set of states
$\endComponent \subseteq \states$ is an \emph{end component} if
\begin{itemize}
\item there exist a player~$\reach$ strategy~$\strategy[\reach]$ and a
  player~$\safe$ strategy~$\strategy[\safe]$ such that for each
  $\state \in \endComponent$, we have
  $\destination(\state,\strategy[\reach],\strategy[\safe])
  \subseteq \endComponent$, and
\item for every pair of states $\state,\state<\prime>
  \in \endComponent$ there is a play $\exampleExecution*$ such that
  $\state[0] = \state$ and $\state[n]=\state<\prime>$ for some $n$,
  and for all $0 \leq i < n$, we have $\state[i] \in \endComponent$
  and it holds $\state[i+1] \in
  \destination(\state[i],\strategy,\strategy*)$.
\end{itemize}
We call an end component~$\endComponent$ \emph{maximal} if there
exists no end component $\endComponent<\prime>$ such that
$\endComponent \subsetneq \endComponent<\prime>$ and \emph{trivial} if
$|\endComponent|=1$.



\section{Value Iteration}
\label{sec:value-iteration}
The idea of value iteration is to assign an initial estimate of the
value to each state and then to successively update it. For standard
value iteration approximating the reachability value from below, the
initial estimates have to be the true values for $\success$ and
$\winning$, i.e. $1$ and $0$, respectively, and below the true values
elsewhere, e.g.~$0$.  Each iteration step propagates the value one
step back further by maximizing the expectation of the value
player~$\reach$ can ensure with respect to the previous estimate.

Formally, we capture estimates as valuations, where a
\emph{valuation}~$\valuation \colon \states \to [0,1]$ is a function
assigning each state~$\state$ a real number $\valuation(\state) \in
[0,1]$ representing the (approximate or true) value of the state. In
addition, let $\valuation, \valuation<\prime>$ be two valuations, we
write $\valuation \leq \valuation<\prime>$ if $\valuation(\state) \leq
\valuation<\prime>(\state)$ for every $\state \in \states$. We can
computed the expected value at a state~$\state$ for a given
valuation~$\valuation$ and strategies~$\strategy$ and~$\strategy*$ by
\[
\preop[\strategy,\strategy*](\valuation)(\state)=
\sum_{\move[\reach],\move[\safe] \in \moves} \sum_{\state<\prime> \in
  \states} \valuation(\state<\prime>) \cdot
\transitions(\state,\move[\reach],\move[\safe])(\state<\prime>) \cdot
\strategy[\reach](\state)(\move[\reach]) \cdot
\strategy[\safe](\state)(\move[\safe]).
\]

\subsection{Lower Bound.}
For the rest of this section, we consider reachability games, where
player~$\reach$ tries to maximize the value. In
\Fref{fig:runningExample}, one can find a concurrent game, which was
originally presented in~\cite{Chatterjee2012}. For this section, we
set $\success = \set{\state[2]}$ and $\fail =
\set{\state[0],\state[1],\state[3],\state[4],\state[5]}$ and let $a$,
$b$ be moves of player~$\reach$ and $c$, $d$ moves of
player~$\safe$. Hence, in \Fref{fig:runningExample},
$\winning=\set{\state[2]}$, which is absorbing.  \cite{Alfaro2004}
presents value iteration from below. We define a slightly simplified
version also used in~\cite{Chatterjee2012}.  In the following, we
denote by $\lowerBound<k>$ the $k$-th iteration of value iteration
from below, where $\lowerBound<k>$ is defined as follows:

\begin{align} 
  \lowerBound<0>(\state) = & \mathsf{if} \ \state \in  \success \ \mathsf{then} \ 1  \ \mathsf{else} \ 0 \\
  \lowerBound<k+1>(\state) = &
  \adjustlimits\sup_{\strategy\in\strategies}\inf_{\strategy*\in\strategies*} \preop[\strategy,\strategy*](\lowerBound<k>)(\state)
\end{align}

\begin{figure}
  \centering
  \begin{tikzpicture}
      \node [style=state-circle] (5) at (-4, 0) {$\state[5]$};
      \node [style=state-circle] (4) at (-2, 0) {$\state[4]$};
      \node [style=state-circle] (3) at (0, 0) {$\state[3]$};
      \node [style=state-circle] (0) at (2, 0) {$\state[0]$};
      \node [style=state-circle] (2) at (4, 1) {$\state[2]$};
      \node [style=state-circle] (1) at (4, -1) {$\state[1]$};

      \draw [style=arrow-segment-after-probabilities,sloped] (4) to node[above]{$\square d$,$1$} (5);
      \draw [style=arrow-segment-after-probabilities,bend right,sloped] (4) to node[below]{$\square c$,$1$} (3);
      \draw [style=arrow-segment-after-probabilities,bend right,sloped] (3) to node[above]{$a \square$,$1$} (4);
      \draw [style=arrow-segment-after-probabilities,sloped] (3) to node[above]{$b \square$,$1$} (0);
      \draw [style=arrow-segment-after-probabilities,sloped] (0) to node[below]{$ac$,$bd$,$1$} (1);
      \draw [style=arrow-segment-after-probabilities,sloped,bend right] (5) to node[below]{$\square\square$,$0.6$} (1);
      \draw [style=arrow-segment-after-probabilities,sloped,bend left] (5) to node[above]{$\square\square$,$0.4$} (2);
      \draw [style=arrow-segment-after-probabilities,sloped, bend left=15] (0) to node[above]{$bc$,$1$} (2);
      \draw [style=arrow-segment-after-probabilities,loop above,sloped] (0) to node[above]{$ad$,$\frac{1}{2}$} (0);
      \draw [style=arrow-segment-after-probabilities,sloped,bend right=15] (0) to node[below]{$ad$,$\frac{1}{2}$} (2);
      \draw [style=arrow-segment-after-probabilities,loop right] (2) to node[right]{$\square\square$} (2);
      \draw [style=arrow-segment-after-probabilities,loop right] (1) to node[right]{$\square\square$} (1);
    \end{tikzpicture}
    \caption{A concurrent game, originally presented
      in~\cite{Chatterjee2012}. $\square$ denotes a move if a player
      only has one available move in a state. }
  \label{fig:runningExample}
\end{figure}
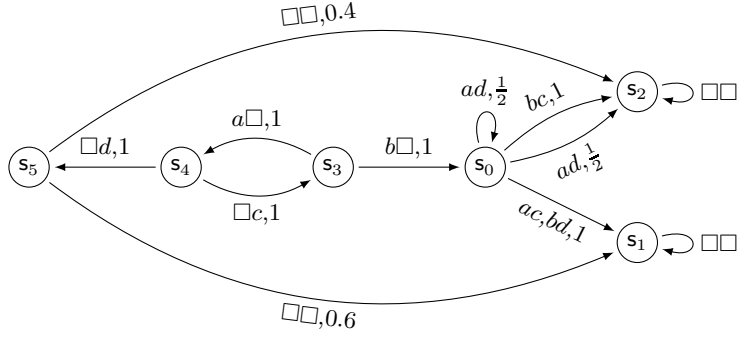

Since $\success$ is absorbing, we have $\lowerBound<k>(\state) = 1$
for all $\state \in \success$ and $\lowerBound<k>(\state) = 0$ for all
$\state \in \winning$ for all $k \in \NN$. 
To compute a monotonically increasing sequence of valuations, we
iterative apply the operator~$\preop$ to the lower bound.  Computing
$\sup_{\strategy\in\strategies}\inf_{\strategy*\in\strategies*}
\preop[\strategy,\strategy*](\lowerBound<k>)(\state)$ corresponds to
solving a one-shot zero-sum matrix game, for instance, for the following payoff
matrix in iteration~0 at state~$\state[0]$:
\[ 
\begin{pmatrix}
  1 & 0 \\ 0 & 1 \\
\end{pmatrix}
\] 
Row~1 corresponds to move~$a$ of player~$\reach$, Row~2 to move~$b$
and Column~1 to move~$c$ and Column~2 to move~$d$ of
player~$\safe$. The content of the matrix represents the payoff
player~$\reach$ achieves. For instance, Cell~1,1 contains the payoff
for player~$\reach$ if player~$\reach$ chooses move~$a$ and
player~$\safe$ chooses move~$c$. Please note that player~$\reach$ can
ensure a value of~$\frac{1}{2}$ in this matrix game by choosing
move~$a$ with probability~$\frac{1}{2}$ and move~$b$ with
probability~$\frac{1}{2}$.

The following theorem states that sequentially updating the value of
the states of a game by solving one-shot matrix games at every state
finally converges to the true reachability value.

\begin{theorem}[Theorem~1 from \cite{Alfaro2004}]
  $\displaystyle{\lim_{k \to \infty}} \lowerBound<k> = \val(\Diamond \success)$
\end{theorem}

Please note that the limit of $\lowerBound<k>$ may not be reached in
finitely many steps since the value may be irrational~\cite{Alfaro2004}.

\subsection{Upper Bound.}
Value iteration from below converges to the value, but at any point in
time we do not know how close we are to the value. To obtain a
stopping criterion, we devise an algorithm approximatingn the value
from above. The distance between the under- and the over-approximation
in a state is then the distance we have at most to the true value.

\paragraph*{Na\"ive Definition} 
Na\"ively, one could define an upper bound iteration
as follows:
\begin{align} 
  \upperBound<0>(\state) = & \mathsf{if} \ \state \in  \winning \ \mathsf{then} \ 0  \ \mathsf{else} \ 1 \\
  \upperBound<k+1>(\state) = & \adjustlimits\sup_{\strategy\in\strategies}\inf_{\strategy*\in\strategies*} \preop[\strategy,\strategy*](\upperBound<k>)(\state) \label{eq:upperBoundRec}
\end{align}

For \CSG{}s, this iteration is a valid over-approximation, but only for
\CSG{}s without \EC{}s, the iteration indeed monotonically converges
to the reachability value from above, which is formalized in
\Fref{theo:BVInoEC}.

\begin{theorem}
  \label{theo:BVInoEC}
  For a \CSG{}~$\game$ without \EC{}s in $\states \setminus
  (\winning \cup \success)$: $\displaystyle{\lim_{k\to\infty}}
  \upperBound<k> = \val(\Diamond \success)$.
\end{theorem}

\emph{Proof Sketch of \Fref{theo:BVInoEC}.} Intuitively, the
over-approximation will be updated from $\success$ and $\winning$
backwards to the states with increasing distance from $\success$ and
$\winning$. In \EC{}-free games, it cannot happen that a set of states
solely depends on each other to determine the value. Hence, the
updates emerging from the correct values of $\success$ and $\winning$
will finally influence the value of all states. We prove the
correctness of this approach by first proving that $\upperBound<k+1>
\leq \upperBound<k>$ for every $k \in \NN$ by a simple induction
over~$k$, which also relies on $\preop$ being monotonic over
valuations. In addition, we prove that $\val(\Diamond \success) \leq
\upperBound<k>$ for all $k \in \NN$. With an argument similar to the
proof of Fixpoint Kleene's Theorem, we show that $\lim_{k \to \infty}
\upperBound<k> = \upperBound<\ast>$ exists and
$\preop[\reach](\upperBound<\ast>) = \upperBound<\ast>$. This suffices
to prove that $\upperBound<\ast>=\val(\Diamond \success)$.

In the presence of non-trivial \EC{}s the above theorem does not hold
since there is no unique fixpoint to the Bellman equations.

\begin{example}
  In \Fref{fig:problemECs}, $\upperBound<0>$ assigns $1$ to both
  $\state[1]$ and $\state[2]$. We have
  $\upperBound<k>(\state[1])=1=\upperBound<k>(\state[2])$ for $k \in
  \NN$ since the strategy, which assigns probability~$1$ to move~$a$
  yields the supremum for $\preop[\reach](\upperBound<k>)(\state[1])$
  and $\preop[\reach](\upperBound<k>)(\state[2])$. However, such a
  strategy yields effectively reachability value~$0$ for both states.
\end{example}

\paragraph*{Bounded Value Iteration.}
Before we present how to overcome the issues of the na\"ive upper
bound iteration, We briefly present the overall bounded value
iteration algorithm. The goal finally is to define a
method~\procname{DEFLATE} such that the algorithm in
\Fref{alg:boundedValueIteration} yields a monotonically decreasing
sequence of valuations over-approximating the reachability value and
converging to it in the limit, which is summerized in
\Fref{theo:upperBoundOptimal}.

\begin{algorithm}
  \begin{algorithmic}[1]
    \Function{\procname{BVI}}{$\epsilon$}
    \State $\mathcal{M} := \procname{MEC}(\game)$
    \Comment{$\mathcal{M}$ is the set of all maximal \EC{}s.}
    \State $k \coloneqq 0$ 
    \Comment{$\lowerBound<0>$ and $\upperBound<0>$ defined as above.}
    
    \DoUntil
    \State $\lowerBound<k+1> \coloneqq  \preop[\reach](\lowerBound<k>)$
    \State $\upperBound<k+1> \coloneqq  \preop[\reach](\upperBound<k>)$
    \For{$\endComponent \in \mathcal{M}$}
    \State \procname{DEFLATE}($\upperBound<k+1>, \endComponent$)
    \EndFor
    \State $k \coloneqq k+1$ 
    \EndDoUntil{$\big(\upperBound(\state[0]) - \lowerBound(\state[0])
      \leq \epsilon\big) \lor \big(\upperBound<k>=\upperBound<k-1>\big) \lor
      \big(\lowerBound<k>=\lowerBound<k-1>\big)$}
    \EndFunction
  \end{algorithmic}
  \caption{Bounded Value Iteration for Concurrent Games}
  \label{alg:boundedValueIteration}
\end{algorithm}

\Fref{alg:boundedValueIteration} depicts bounded value iteration,
i.e. the parallel computation of the upper and lower bound to bound
the distance to the true value. If these approximations are closer
than $\epsilon$, we know that both approximations are at most
$\epsilon$-away from the real value.

\begin{theorem}[Optimality of Upper Bound]
  \label{theo:upperBoundOptimal}
  \begin{enumerate}
  \item If $\upperBound<k> = \upperBound<k+1>$ for some $k \in \NN$ in
    \Fref{alg:boundedValueIteration}, then $\upperBound<k> =
    \val(\Diamond \success)$.
  \item $\displaystyle{\lim_{k \to \infty}} \upperBound<k> =
    \val(\Diamond \success)$ in
    \Fref{alg:boundedValueIteration}.
  \end{enumerate}
\end{theorem}

\emph{Proof Sketch of \Fref{theo:upperBoundOptimal}.} To formally
prove the claim, we show that \procname{DEFLATE} is also
monotone. Then, we can show that $\upperBound<k+1> \leq
\upperBound<k>$ for all $k \in \NN$. The rest of the proof does not
differ from the proof for games without non-trivial end components,
i.e. we show that $\upperBound<\ast>$ is a unique fixpoint of the
updates to $\upperBound<k>$ in \Fref{alg:boundedValueIteration}, which
suffices to show that $\upperBound<\ast>$ is indeed $\val[\reach]$.

\begin{figure}
  \centering
  \begin{tikzpicture}
      \node [style=state] (1) at (0, 0) {$\state[1]$};
      \node [style=state] (2) at (-2, 0) {$\state[2]$};
      \node [style=state] (3) at (2, 0.5) {$\state[3]$};
      \node [style=state] (4) at (2, -0.5) {$\state[4]$};
      \node [style=state] (0) at (-4, 0) {$\state[0]$};
  
      \draw [style=arrow-segment-after-probabilities,bend right,sloped] (1) to node[above]{$a\square$,$1$} (2);
      \draw [style=arrow-segment-after-probabilities,bend right,sloped] (2) to node[below]{$a\square$,$1$} (1);
      \draw [style=arrow-segment-after-probabilities,bend right,sloped] (0) to node[below]{$a\square$,$\frac{1}{2}$} (2);
      \draw [style=arrow-segment-after-probabilities,bend right,sloped] (2) to node[above]{$b\square$,$1$} (0);
      \draw [style=arrow-segment-after-probabilities,sloped] (1) to node[above]{$b\square$,$\frac{1}{2}$} (3);
      \draw [style=arrow-segment-after-probabilities,sloped] (1) to node[below]{$b\square$,$\frac{1}{2}$} (4);
      \draw [style=arrow-segment-after-probabilities,loop right] (3) to node[right]{$a\square$,$1$} (3);
      \draw [style=arrow-segment-after-probabilities,loop right] (4) to node[right]{$a\square$,$1$} (4);
      \draw [style=arrow-segment-after-probabilities,loop left] (0) to node[left]{$a\square$,$\frac{1}{2}$} (0);
    \end{tikzpicture}
    \caption{We set $\success = \set{\state[3]}$ and $\fail =
      \set{\state[0],\state[1],\state[2],\state[4]}$. All states are fully
      controlled by player~$\reach$. Both $\state[1]$ and $\state[2]$
      are not part of $\winning$. Hence, $\upperBound<0>$ assigns $1$
      to both states. Since $1$ is larger than $0.5$,
      $\upperBound<k+1>$ still assigns $1$ for all $k \in \NN$ with an
      optimal strategy always preferring move~$a$ over move~$b$ or
      any non-Dirac distribution over both. However, the value this
      strategy yields will effectively be~$0$ since we never visit a
      state in~$\success$.}
  \label{fig:problemECs}
\end{figure}
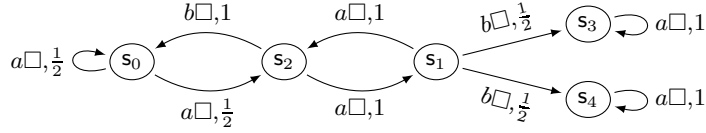

\paragraph*{Theoretical Foundation of Deflating.}
There are two observations, which are crucial for deflating: (1) A
state in an \EC{} cannot have a better reachability value than it
achieves by leaving the \EC{} since staying in an \EC{} outside of
$\success$ will effectively yield value~$0$. (2) The states in the end
component may promise each other unsubstantiated overly high
reachability values. 

Such a problem occurs, for instance, in the
\EC{}~$\set{\state[1],\state[2]}$ in \Fref{fig:problemECs}. If we
initialize all states except~$\state[4]$ with estimate~$1$, states
$\state[1]$ and $\state[2]$ will always promise each other value~$1$
although none of the states can really achieve it.

This process of adjusting the value in \EC{}s~is called
\emph{deflating}~\cite{Kelmendi2018}. In more detail, we will reduce
the estimate of the reachability value in end components to the best
estimate they can achieve when forced to leave. Here, we define
whether a player stays or leaves the end component over its potential
successors.

For an end component~$\endComponent$ and a $\state \in \endComponent$,
we call a move~$\move[\reach] \in \moveAssignment[\reach](\state)$
\emph{staying} if $\forall \move[\safe] \in
\moveAssignment[\safe](\state):
\destination(\state,\move[\reach],\move[\safe])
\subseteq \endComponent$ and \emph{leaving} if $\forall \move[\safe]
\in \moveAssignment[\safe](\state):
\destination(\state,\move[\reach],\move[\safe])
\not\subseteq \endComponent$

One can observe that single moves can be neither staying nor leaving
in concurrent games. In turn-based games, the definitions of staying
and leaving moves are complementary since each state and thus, every
transition, is controlled by a single player.

\begin{example}
  Consider for instance move~$a$ at state~$\state[1]$ for the
  \EC{}~$\set{\state[1],\state[2]}$ in
  \Fref{fig:leaving}. Player~$\reach$ can neither enforce to stay in
  the \EC\ nor can player~$\reach$ enforce to leave it. The
  state~$\state[5]$ in \Fref{fig:leaving} is an example of a state
  (and an end component), which does not have any move~$\move[\reach]$
  for player~$\reach$ such that for all moves~$\move[\safe]$ of
  player~$\safe$ holds
  $\destination(\state[5],\move[\reach],\move[\safe]) \not\subseteq
  \set{\state[5]}$. However, the strategy, which assigns
  probability~$\frac{1}{2}$ to both available moves ensures that
  states outside~$\set{\state[5]}$ are seen with positive probability.
\end{example}

To overcome this issue, we cannot simply restrict player~$\reach$ to
strategies that sign a positive probability to moves, which lead to
states outside of the \EC{} with a positive probability since the
limit of a sequence of such strategies might not satisfy the
property. This is a difficulty for the computation of~$\preop$.
Formally, for an \EC{}~$\endComponent$, player $\reach$ and $\state
\in \endComponent$, we denote by

\begin{align*}
  \strategies[\reach]<\exit[\reach]<\endComponent>>(\state) \coloneqq
  & \set{ \strategy[\reach] \in \strategies[\reach] \mid \forall
    \strategy[\safe] \in \strategies[\safe].\ \destination(\state,
    \strategy[\reach], \strategy[\safe]) \not\subseteq \endComponent
    \\ & \land \nexists \move[\reach] \in
    \support(\strategy[\reach](\state)): \forall \move[\safe] \in
    \moveAssignment[\safe](\state):
    \destination(\state,\move[\reach],\move[\safe])
    \subseteq \endComponent}
\end{align*}
    
the set of strategies, which force the play to leave $\endComponent$
from $\state$, while not using any staying move. Now, we extend the
pre-operator as follows:

\begin{equation}
\preop<\exit[\reach]<\endComponent>>(\valuation)(\state) \coloneqq
\adjustlimits\sup_{\strategy[\reach] \in
  \strategies[\reach]<\exit[\reach]<\endComponent>>(\state)}
\inf_{\strategy[\safe] \in \strategies[\safe]}
\preop[\strategy[\reach], \strategy[\safe]](\valuation)(\state)
\end{equation}

We denote by $\best<\exit[\reach]<\endComponent>>(\valuation)(\state)
\in \strategies[\reach]<\exit[\reach]<\endComponent>>(\state)$ one
optimal strategy of the modified one-shot matrix game, which considers
leaving strategies only. Such a strategy exists since we only
consider end components not in $\reach$ or $\winning$ and a end
component without such a strategy is part of $\winning$. We define the
\emph{best exit} of an end component~$\endComponent$ for player~$i$
with respect to a valuation~$\valuation$ by

\begin{equation*}
  \bestExit[\reach]<\valuation>(\endComponent) \coloneqq \max_{\state \in \endComponent} \preop<\exit[\reach]<\endComponent>>(\valuation)(\state)
\end{equation*}

\paragraph*{Algorithmically Deflating.}
We finally can devise an algorithm for \procname{DEFLATE}. First of
all, please note that we can compute
$\bestExit[\reach]<\valuation>(\endComponent)$ by removing moves of
player~$\reach$, which surely stay inside the end component, and by
constraining the solutions of the linear optimization problem to solve
to such solutions assigning a probability greater than~$0$ to states
outside the end component~$\endComponent$.

\begin{figure}
  \centering
  \begin{tikzpicture}
      \node [style=state] (1) at (0, 0) {$\state[1]$};
      \node [style=state] (2) at (-2, 0) {$\state[2]$};
      \node [style=state] (3) at (2, 0.5) {$\state[3]$};
      \node [style=state] (4) at (2, -0.5) {$\state[4]$};
      \node [style=state] (5) at (4,0) {$\state[5]$};
  
      \draw [style=arrow-segment-after-probabilities,bend right,sloped] (1) to node[above]{$ac$} (2);
      \draw [style=arrow-segment-after-probabilities,bend right,sloped] (2) to node[below]{$a\square$} (1);
      \draw [style=arrow-segment-after-probabilities,sloped] (1) to node[above]{$ac$,$\frac{1}{2}$} (3);
      \draw [style=arrow-segment-after-probabilities,sloped] (1) to node[below]{$ac$,$\frac{1}{2}$} (4);
      \draw [style=arrow-segment-after-probabilities,loop above] (3) to node[above]{$a\square$} (3);
      \draw [style=arrow-segment-after-probabilities,loop below] (4) to node[below]{$a\square$} (4);
      \draw [style=arrow-segment-after-probabilities,loop above] (5) to node[above]{$ac$,$bd$} (5);
      \draw [style=arrow-segment-after-probabilities,sloped] (5) to node[above]{$bc$,$ad$,$\frac{1}{2}$} (3);
      \draw [style=arrow-segment-after-probabilities,sloped] (5) to node[below]{$bc$,$ad$,$\frac{1}{2}$} (4);
    \end{tikzpicture}
    \caption{We set $\success = \set{\state[3]}$, all other states are
      in $\fail$. $ab$, for instance, denotes that player~$\reach$
      plays move~$a$ and player~$\safe$ plays move~$b$. The move~$a$
      at state~$\state[1]$ is neither leaving nor staying
      for \EC{}~$\set{\state[1],\state[2]}$ since the behavior of $a$
      depends on the move chosen by player~$\safe$. While
      player~$\reach$ can neither ensure to leave~$\set{\state[5]}$ by
      move~$a$ nor by~$c$, but the strategy, which
      assigns~$\frac{1}{2}$ to both moves, ensures leaving
      $\set{\state[5]}$ with positive probability.}
  \label{fig:leaving}
\end{figure}
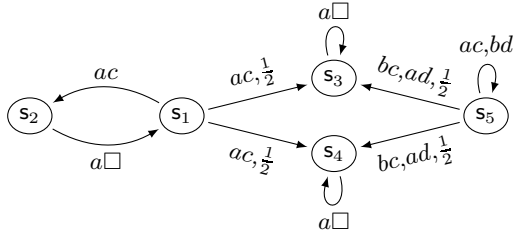

Once we now the best exit of each state, we use the attractor
construction to compute the set of states, which can ensure to visit
the states with the best exit of an end component~$\endComponent$ as
follows: Please note that we use the computation of the attractor as
in~\cite{Alfaro2000}. Let $B \subseteq \endComponent$, then

\begin{align*}
    \attractor<0>(B) \coloneqq & B \\
    \attractor<k+1>(B) \coloneqq & \attractor<k>(B) \cup \\
        & \set{\state \in \endComponent \mid
    \exists \move[\reach] \in \moveAssignment[\reach](\state): \forall \move[\safe] \in
    \moveAssignment[\safe](\state):
    \destination(\state,\move[\reach],\move[\safe]) \subseteq \attractor<k>(B)}
\end{align*}

The computation will clearly terminate after at most $|\endComponent|$
iterations. Therefore, we set $\attractor(B) \coloneqq
\attractor<|\endComponent|>(B)$. This corresponds to the set of the
states, for which player~$\reach$ surely reaches $B$.\footnote{The
  results in the subsequent sections also hold if we compute the set
  of states, which reaches the set guaranteeing the best exit
  almost-surely.} This finally leads to the algorithm for deflating
presented in~\Fref{alg:deflate}. \procname{DEFLATE} first computes
the attractor of the best exit, then updates all states in the end
component and finally, removes all states from the previously computed
attractor. This process is iterated until there is no state left. This
intuitively leads to updates from the best to the worst best exit
(w.r.t. the current iteration) a player can enforce in state.

\begin{algorithm}
  \begin{algorithmic}[1]
    \Function{\procname{DEFLATE}}{$\upperBound<k+1>, \anySet$}
    \DoUntil
    \State $B \coloneqq \attractor(\set{\state \in \anySet \mid
      \preop<\exit[\reach]<\anySet>>(\upperBound<k+1>)(\state) 
      = \bestExit[\reach]<\upperBound<k+1>>(\anySet)})$
    \For{$\state \in \anySet$}
    \State $\upperBound<k+1>(\state) \coloneqq 
    \min\big(\upperBound<k+1>(\state),\bestExit[\reach]<\upperBound<k+1>>(\anySet)\big)$
    \EndFor
    \State $\anySet \coloneqq \anySet \setminus B$
    \EndDoUntil{$\anySet = \emptyset$}
    \EndFunction
  \end{algorithmic}
  \caption{Update upper bound of a single \MEC{}.}
  \label{alg:deflate}
\end{algorithm}


\begin{example}
    \label{ex:runningExampleVI}

    \begin{table}[ht!]
        \centering\footnotesize
        \begin{tabular}[ht!]{c|cccccc}
            $k$ $\backslash$ state & $\state[0]$ & $\state[1]$ & $\state[2]$ & $\state[3]$ & $\state[4]$ & $\state[5]$\\\hline
            0      & 0             & 0 & 1 & 0  & 0   & 0 \\
            1      & 0.375         & 0 & 1 & 0   & 0   & 0.4\\ \\
            2      & 0.40741       & 0 & 1 & 0.4 & 0.4 & 0.4\\
            3      & 0.41304       & 0 & 1 & 0.4 & 0.4 & 0.4\\
            \dots \\
            $\ast$ &$\sqrt{2}-1$  & 0 & 1 & 0.4 & 0.4 & 0.4\\
        \end{tabular}
        \hfill
        \begin{tabular}[ht!]{c|cccccc}
            $k$ $\backslash$ state & $\state[0]$ & $\state[1]$ & $\state[2]$ & $\state[3]$ & $\state[4]$ & $\state[5]$\\\hline
            0        & 1 & 0 & 1 & 1 & 1   & 1 \\
            1        & 0.5 & 0 & 1 & 1 & 1   & 0.4\\
 \procname{DEFLATE}  & 0.5 & 0 & 1 & 0.4 & 0.4   & 0.4\\
            2        & 0.43 & 0 & 1 & 0.4 & 0.4 & 0.4\\
            3        & 0.42 & 0 & 1 & 0.4 & 0.4 & 0.4\\
            \dots \\
            $\ast$ & $\sqrt{2}-1$ & 0 & 1 & 0.4 & 0.4 & 0.4\\
        \end{tabular}
        \caption{Lower Bound~\VI{} on the left and Upper Bound~\VI{} on the right for
        the Game in \Fref{fig:runningExample}, where we approximate the
        value for state $\state[0]$ with decimals.}
        \label{tab:runningExample}
    \end{table}

    In \Fref{tab:runningExample} we apply bounded value iteration on
    the game in \Fref{fig:runningExample}.  We present the lower bound
    iteration on the left and the upper bound iteration on the
    right. Since the value of $\state[0]$ is irrational, it is not
    reached within finitely many steps. However, in this example, we
    need only three steps to approximate it with precision~$0.01$.
    Without deflating the upper bound after Iteration~1, the upper
    bound for $\state[3]$ would always be determined by the upper
    bound of $\state[0]$. Yet, if player~$\reach$ decides to always
    play~$a$, then staying in the end component
    $\set{\state[3],\state[4]}$ will yield the value~0. Hence, we must
    rather take into account the leaving action $b$ from $\state[3]$
    which yields the true value of $0.4$. This reasoning will be more
    apparent in \Fref{ex:runningExampleSI}, where we also present the
    respective strategies for player $\safe$.  Once we have deflated
    the end component, all values remain constant except for that of
    $\state[0]$ which approaches $\sqrt{2}-1$.
\end{example}


\section{Strategy Iteration}
\label{sec:strategy-iteration}
In the previous section, we presented an algorithm for VI that can
provide both upper and lower bounds on the value, which converge to
the actual value, at any point in time. Another popular approach for
solving games is SI. So far there is no way of telling how close we
have approximated the true value for general concurrent games. For the
lower bound, convergence results exist~\cite{Chatterjee2013}. For the
upper bound, however, the only results so far are for the special case
of turn-based stochastic games. The problem with convergence of the
upper bound is the same as in the case of VI, namely mistakenly
overestimating the value within end components and thus not leaving
them. We deal with end components by \emph{deflating} them to a safe
over approximation that takes into account leaving strategies.

For SI from below, we iteratively improve a given strategy for player
$\reach$. Note that for a given $\strategy[\reach] \in
\strategies[\reach]$ the value
$\val[\reach:\strategy[\reach]](\Diamond \success)$ always provides a
lower bound to the true value $\val[\reach](\Diamond \success) =
\sup_{\strategy[\reach]<\prime> \in
  \strategies[\reach]}\val[\reach:\strategy[\reach]<\prime>](\Diamond
\success)$. Therefore, it is not clear how to come up with an upper
bound, given only a strategy for player $\reach$. The key is to
consider strategies for player $\safe$, as well. Using a similar
argument, a fixed strategy $\strategy[\safe] \in \strategies[\safe]$
always provides a lower bound $\val[\safe:\strategy[\safe]](\Box
\fail) \leq \sup_{\strategy[\safe]<\prime> \in
  \strategies[\safe]}\val[\safe:\strategy[\safe]<\prime>](\Box \fail)
= \val[\safe](\Box \fail)$. Since $\val[\reach] = 1 - \val[\safe]$, we
can compute an upper bound for player $\reach$ from a lower bound for
player $\safe$. Taking this discussion into account, the bounded SI
algorithm works essentially the same as that for bounded VI.

\begin{algorithm}
  \begin{algorithmic}[1]
    \Require concurrent stochastic game $\exampleGame$ with reach set~$\success$
    \Ensure memoryless strategies~$\strategy[\reach], \strategy[\safe]$
    \Function{\procname{SI}}{}
      \State Compute $\winning[\safe] = \set{\state \in \states \mid
      \val[\safe](\Box \fail)(\state)=1}$
      \State Compute the set of all \MEC{}s $\mathcal{M}$.
      \State Let $\strategy[\reach]<0> \in \strategies[\reach], \strategy[\safe]<0> \in \strategies[\safe]$ be arbitrary
        memoryless strategies and let $k=0$.
      \Statex
      \DoUntil 
        \State $\lowerBound<k> \coloneqq \val[\reach:\strategy[\reach]<k>](\Diamond \success)$
        \State $\upperBound<k> \coloneqq 1 - \val[\safe:\strategy[\safe]<k>](\Box \fail)$
        \For{$\endComponent \in \mathcal{M}$}
          \State \procname{DEFLATE($\strategy[\safe]<k>$,$\upperBound<k>$,$\endComponent$)}
        \EndFor
        \Statex
        \State $\lowerSet<k> \coloneqq \set{\state \in \states \setminus (\winning[\safe]
          \cup \success) \mid \preop[\reach](\lowerBound<k>)(\state) \neq
          \lowerBound<k>(\state)}$
        \State $\upperSet<k> \coloneqq \set{\state \in \states \setminus (\winning[\safe]
          \cup \success) \mid \preop[\safe](1-\upperBound<k>)(\state) \neq
          1-\upperBound<k>(\state)}$
        \State Compute $\strategy[\reach]<\ast> \in \strategies[\reach]$ s.t. for
          $\state \in \lowerSet<k>$ holds
          $\preop[\reach:\strategy[\reach]<\ast>](\lowerBound<k>)(\state) =
          \preop[\reach](\lowerBound<k>)(\state)$
        \State Compute $\strategy[\safe]<\ast> \in \strategies[\safe]$ s.t. for
          $\state \in \upperSet<k>$ holds
          $\preop[\safe:\strategy[\safe]<\ast>](1-\upperBound<k>)(\state) =
          \preop[\safe](1-\upperBound<k>)(\state)$
        \Statex
        \State Define $\strategy[\reach]<k+1>$ as follows for each state
          $\state \in \states$: $\strategy[\reach]<k+1>(\state) \coloneqq 
          \begin{cases}
            \strategy[\reach]<k>(\state) & \state \not\in \lowerSet<k>\\
            \strategy[\reach]<\ast>(\state) & \state \in \lowerSet<k>
          \end{cases}$
        \State Define $\strategy[\safe]<k+1>$ as follows for each state
          $\state \in \states$: $\strategy[\safe]<k+1>(\state) \coloneqq 
          \begin{cases}
            \strategy[\safe]<k>(\state) & \state \not\in \upperSet<k>\\
            \strategy[\safe]<\ast>(\state) & \state \in \upperSet<k>
          \end{cases}$
        \State $k \coloneqq k+1$
      \EndDoUntil ~$\lowerSet<k> = \emptyset$ or $\upperSet<k> = \emptyset$ or $|\upperBound<k> - \lowerBound<k>| < \epsilon$
      \Statex
      \State \Return $\strategy[\reach]<k>, \strategy[\safe]<k>$
    \EndFunction
  \end{algorithmic}
  \caption{Strategy Iteration for Concurrent Games}
   \label{alg:strategyIteration}
\end{algorithm}

\begin{algorithm}
  \begin{algorithmic}[1]
    \Function{\procname{DEFLATE}}{$\strategy[\safe]<k+1>$,$\upperBound<k+1>$,$\endComponent$}
      \DoUntil
        \State $B \coloneqq \attractor(\set{\state \in \endComponent
          \mid
          \preop<\exit(\endComponent)>(\upperBound<k+1>)(\state)
          = \bestExit<\upperBound<k+1>>(\endComponent)})$
        \State Let $\strategy[\reach] \in \strategies[\reach]<\exit(\endComponent)>$ s.t. 
        $\preop[\reach:\strategy[\reach]](\upperBound<k+1>)(\state)
          = \bestExit<\upperBound<k+1>>(\endComponent)$ for all
          $\state \in B$.
        \State Let $\strategy[\safe]$
        s.t. $\preop[\strategy[\reach],\strategy[\safe]](\upperBound<k+1>)(\state)
        = \bestExit<\upperBound<k+1>>(\endComponent)$ for all
        $\state \in B$
        \For{$\state \in \endComponent$}
          \State $\upperBound<k+1>(\state) \coloneqq
          \min(\upperBound<k+1>(\state),\bestExit<\upperBound<k+1>>(\endComponent))$
          \State $\strategy[\safe]<k+1>(\state) \coloneqq \strategy[\safe](\state)$
        \EndFor
        \State $\endComponent := \endComponent \setminus B$
      \EndDoUntil{$\endComponent = \emptyset$}
    \EndFunction
  \end{algorithmic}
  \caption{Algorithm to Deflate the Safety Value in \MEC{}s for
    player~$\safe$}
  \label{alg:inflating}
\end{algorithm}

The computation of the lower bound in \Fref{alg:strategyIteration}
corresponds to the standard SI algorithm for which convergence is
known~\cite{Chatterjee2013}. We start with an arbitrary strategy for
player $\reach$. In the set $\lowerSet<k>$ we store the states that
currently underestimate the value. For those, we update the strategy
such that it optimizes for the current lower bound. The computation of
the upper bound is analogous except for the additional call to
\procname{DEFLATE}. Just as in bounded value iteration,
\procname{DEFLATE} reduces the upper bound. For this, it computes the
optimal player $\reach$ strategy (w.r.t the current upper bound) that
leaves the end component. Then, the player $\safe$ strategy is
adjusted, to be the best response to such a leaving strategy of player
$\reach$. We stop whenever one of the sets $\lowerSet<k>$ or
$\upperSet<k>$ is empty or the difference between the lower and the
upper bound is sufficiently small.

When comparing bounded strategy iteration to bounded value iteration,
essentially the only difference is that we keep track of the
strategies that are used to attain the current estimate of the
respective bounds. Apart from that, most of the computations are
analogous to those in BVI. However, note that the computation of
$\upperBound<k>$ and $\lowerBound<k>$ are quite different in that they
are computed as the true reachability value for fixed strategies. In
contrast, BVI computes the bounds by means of the
$\preop$-operator. Intuitively, there is not much difference between
the two approaches as the $\preop$-operator computes the true value in
the long run. Since we have already proven the correctness of BVI,
proving the correctness of BSI amounts to showing that the two really
behave the same in the long run, which given the similarities is not
too difficult.

The correctness of \Fref{alg:strategyIteration} follows the proof of correctness for Algorithm~1 of \cite{Chatterjee2013}, which relies on the existence of a matching value iteration algorithm.

\begin{theorem}
  \label{theo:siOptimal}
  \[\val(\Box \fail) = \val[\safe:\strategy[\safe]<\ast>](\Box
  \fail) \] where $\lim_{i\to\infty} \strategy[\safe]<i> =
  \strategy[\safe]<\ast>$ by \Fref{alg:strategyIteration}.
\end{theorem}

\emph{Proof Sketch.} We prove \Fref{theo:siOptimal} by an induction over $k \in \NN$, which shows that $1-\upperBound<k> \leq \valuation<k> \leq \val(\Box \fail)$, where the last inequality trivially holds since no strategy can provide a better value than the actual value. For the proof, we use that \procname{DEFLATE} is monotone and that all updates to $\strategy[\safe]<k>$ just happen to indeed provide the valuation~$\valuation<k>$.  

Since the upper bound is computed as the complement to the safe value, the above theorem implies that the upper bound converges to the true value.

\begin{example}
    \label{ex:runningExampleSI}

\begin{table}[ht!]
  \centering\footnotesize
  \begin{tabular}[ht!]{c|cccccc}
    $k$ $\backslash$ state & $\state[0]$ & $\state[4]$ \\\hline
    0      & c: 0.5 & c: 0.5 \\
    1      & c: $0.4$ & c: 0 \\
    2      & c: $0.41$ & c: 0 \\
  \end{tabular}
  \hfill
   \begin{tabular}[ht!]{c|cccccc}
    $k$ $\backslash$ state & $\state[0]$ & $\state[1]$ & $\state[2]$ & $\state[3]$ & $\state[4]$ & $\state[5]$\\\hline
    0      & 0.33 & 0 & 1 & 0.33 & 0.37   & 0.40 \\
    1      & 0.40 & 0 & 1 & 0.40 & 0.40   & 0.40\\
    2      & 0.41 & 0 & 1 & 0.40 & 0.40 & 0.40\\
  \end{tabular}

  \centering\footnotesize
  \begin{tabular}[ht!]{c|cccccc}
    $k$ $\backslash$ state & $\state[0]$ & $\state[3]$ \\\hline
    0      & a: 0.5 & a: 0.5 \\
    \procname{DEFLATE} & a: 0.5 & a: 1 \\
    1      & a: $0.57$ & a: 1 \\
    2      & a: $0.58$ & a: 1 \\
  \end{tabular}
  \hfill
   \begin{tabular}[ht!]{c|cccccc}
    $k$ $\backslash$ state & $\state[0]$ & $\state[1]$ & $\state[2]$ & $\state[3]$ & $\state[4]$ & $\state[5]$\\\hline
    0      & 0.50 & 0 & 1 & 0.50 & 0.50   & 0.40 \\
    \procname{DEFLATE}     & 0.50 & 0 & 1 & 0.40 & 0.40   & 0.40 \\
    1      & 0.43 & 0 & 1 & 0.40 & 0.40   & 0.40\\
    2      & 0.42 & 0 & 1 & 0.40 & 0.40 & 0.40\\
  \end{tabular}

  \caption{Strategy iteration for player $\reach$ and player $\safe$ on the top resp bottom. Strategies on the left, corresponding values on the right.}
  \label{tab:runningExampleSI}
\end{table}

    Consider again \Fref{fig:runningExample}. In \Fref{tab:runningExampleSI} we show the strategies and corresponding values for both players. The strategies are only given for the states where the choices are non-trivial for the respective player. Since each player has only two actions to choose from, we show the probability assignment for only one of the actions, from which the assignment for the other action is straighforward to compute. 

    In \cite{Chatterjee2012} the authors explain why strategy iteration from below for the safety player does not converge for this game.\footnote{Note that the game they present is not precisely the same. Concretely, on action pair $ad$ from $\state[0]$ we go to $\state[2]$ with probability $1/2$ whereas in \cite{Chatterjee2012} we go to $\state[1]$ instead. However, the one-shot matrix the authors present corresponds to our version rather than to theirs. Consequently, their argument applies to our game and does not apply to theirs.} If player $\reach$ plays $c$ from $\state[4]$, then for player $\safe$ it seems as if it did not make a difference whether to play $a$ or $b$ from $\state[3]$ as both seem to realize the same value---namely that of $\state[0]$. In fact, the best response for player $\safe$ at $\state[3]$ is to play $a$ in which case player $\reach$ would attain the value 0 when staying in the end component $\set{\state[3],\state[4]}$. When computing $\preop{}$ on the upper bound, this fact is not properly reflected and therefore strategy iteration for the safety player does not converge to the true value at states $\state[3]$ and $\state[4]$. In step $0'$, we correct this by calling \procname{DEFLATE} and thus taking into account only exiting strategies from $\state[4]$. Having done so, player $\safe$ realizes that the reasonable choice at $\state[3]$ is to play $a$ rather than $b$.
\end{example}


\section{Conclusion and Future Work}
\label{sec:conclusion}
We have provided the first stopping criterion for both value and
strategy iteration on concurrent games with reachability and safety
objectives as well as anytime algorithms with the bounds on the
current error. Since the games are concurrent and since ($\epsilon$-)optimal strategies may need to be
randomized, we could not use the technique of simple end
components of \cite{Kelmendi2018}. 
Instead, we iteratively update maximal end
components and deflate only those states, which can ensure the currently best
exiting combination of moves. 
We leave an efficient implmentation for future work, as an extension---similarly to \cite{Baier17,Kelmendi2018}---of the standard model checker \prism{}-games~\cite{Kwiatkowska2018}.



\newpage
\addcontentsline{toc}{chapter}{Bibliography}
\bibliography{Quellen,bibliography}

\begin{thebibliography}{10}

\bibitem{Baier17}
Christel Baier, Joachim Klein, Linda Leuschner, David Parker, and Sascha
  Wunderlich.
\newblock Ensuring the reliability of your model checker: Interval iteration
  for markov decision processes.
\newblock In {\em {CAV} {(1)}}, volume 10426 of {\em Lecture Notes in Computer
  Science}, pages 160--180. Springer, 2017.

\bibitem{Brazdil2014}
Tom{\'a}{\v{s}} Br{\'a}zdil, Krishnendu Chatterjee, Martin Chmelik,
  Vojt{\v{e}}ch Forejt, Jan K{\v{r}}et{\'\i}nsk{\`y}, Marta Kwiatkowska, David
  Parker, and Mateusz Ujma.
\newblock Verification of {M}arkov decision processes using learning
  algorithms.
\newblock In {\em International Symposium on Automated Technology for
  Verification and Analysis}, pages 98--114. Springer, 2014.

\bibitem{Chatterjee2009}
Krishnendu Chatterjee, Luca~de Alfaro, and Thomas~A Henzinger.
\newblock Termination criteria for solving concurrent safety and reachability
  games.
\newblock In {\em Proceedings of the twentieth annual ACM-SIAM symposium on
  Discrete algorithms}, pages 197--206. SIAM, 2009.

\bibitem{Chatterjee2006}
Krishnendu Chatterjee, Luca de~Alfaro, and Thomas~A Henzinger.
\newblock Strategy improvement for concurrent reachability games.
\newblock In {\em null}, pages 291--300. IEEE, 2006.

\bibitem{Chatterjee2012}
Krishnendu Chatterjee, Luca de~Alfaro, and Thomas~A Henzinger.
\newblock Strategy improvement for concurrent reachability and safety games.
\newblock {\em arXiv preprint arXiv:1201.2834}, 2012.

\bibitem{Chatterjee2013}
Krishnendu Chatterjee, Luca de~Alfaro, and Thomas~A Henzinger.
\newblock Strategy improvement for concurrent reachability and turn-based
  stochastic safety games.
\newblock {\em Journal of computer and system sciences}, 79(5):640--657, 2013.

\bibitem{Condon90}
Anne Condon.
\newblock On algorithms for simple stochastic games.
\newblock In {\em Advances In Computational Complexity Theory}, volume~13 of
  {\em {DIMACS} Series in Discrete Mathematics and Theoretical Computer
  Science}, pages 51--72. {DIMACS/AMS}, 1990.

\bibitem{alfarothesis}
Luca De~Alfaro.
\newblock How to specify and verify the long-run average behaviour of
  probabilistic systems.
\newblock In {\em Proceedings. Thirteenth Annual IEEE Symposium on Logic in
  Computer Science (Cat. No. 98CB36226)}, pages 454--465. IEEE, 1998.

\bibitem{Alfaro2000}
Luca de~Alfaro and Thomas~A Henzinger.
\newblock Concurrent omega-regular games.
\newblock In {\em Logic in Computer Science, 2000. Proceedings. 15th Annual
  IEEE Symposium on}, pages 141--154. IEEE, 2000.

\bibitem{AlfaroHK98}
Luca de~Alfaro, Thomas~A. Henzinger, and Orna Kupferman.
\newblock Concurrent reachability games.
\newblock In {\em {FOCS}}, pages 564--575. {IEEE} Computer Society, 1998.

\bibitem{Alfaro2004}
Luca de~Alfaro and Rupak Majumdar.
\newblock Quantitative solution of omega-regular games.
\newblock {\em Journal of Computer and System Sciences}, 68(2):374 -- 397,
  2004.
\newblock Special Issue on STOC 2001.

\bibitem{storm}
Christian Dehnert, Sebastian Junges, Joost{-}Pieter Katoen, and Matthias Volk.
\newblock A storm is coming: {A} modern probabilistic model checker.
\newblock In {\em {CAV} {(2)}}, volume 10427 of {\em Lecture Notes in Computer
  Science}, pages 592--600. Springer, 2017.

\bibitem{Etessami2006}
Kousha Etessami and Mihalis Yannakakis.
\newblock Recursive concurrent stochastic games.
\newblock In Michele Bugliesi, Bart Preneel, Vladimiro Sassone, and Ingo
  Wegener, editors, {\em Automata, Languages and Programming}, pages 324--335,
  Berlin, Heidelberg, 2006. Springer Berlin Heidelberg.

\bibitem{Everett1957}
H.~Everett.
\newblock {\em RECURSIVE GAMES}, pages 47--78.
\newblock Princeton University Press, 1957.

\bibitem{Haddad2018}
Serge Haddad and Benjamin Monmege.
\newblock Interval iteration algorithm for mdps and imdps.
\newblock {\em Theoretical Computer Science}, 735:111--131, 2018.

\bibitem{Kelmendi2018}
Edon Kelmendi, Julia Kr{\"a}mer, Jan K{\v{r}}et{\'i}nsk{\'y}, and Maximilian
  Weininger.
\newblock Value iteration for simple stochastic games: Stopping criterion and
  learning algorithm.
\newblock In Hana Chockler and Georg Weissenbacher, editors, {\em Computer
  Aided Verification}, pages 623--642, Cham, 2018. Springer International
  Publishing.

\bibitem{Kwiatkowska2018}
Marta Kwiatkowska, Gethin Norman, David Parker, and Gabriel Santos.
\newblock Automated verification of concurrent stochastic games.
\newblock In {\em International Conference on Quantitative Evaluation of
  Systems}, pages 223--239. Springer, 2018.

\bibitem{McMahan2005}
H~Brendan McMahan, Maxim Likhachev, and Geoffrey~J Gordon.
\newblock Bounded real-time dynamic programming: Rtdp with monotone upper
  bounds and performance guarantees.
\newblock In {\em Proceedings of the 22nd international conference on Machine
  learning}, pages 569--576. ACM, 2005.

\bibitem{Parthasarathy1973}
T.~Parthasarathy.
\newblock Discounted, positive, and noncooperative stochastic games.
\newblock {\em International Journal of Game Theory}, 2(1):25--37, Dec 1973.

\bibitem{Puterman}
Martin~L Puterman.
\newblock {M}arkov decision processes: Discrete stochastic dynamic programming.
\newblock 1994.

\bibitem{QuatmannK18}
Tim Quatmann and Joost{-}Pieter Katoen.
\newblock Sound value iteration.
\newblock In {\em {CAV} {(1)}}, volume 10981 of {\em Lecture Notes in Computer
  Science}, pages 643--661. Springer, 2018.

\end{thebibliography}

\newpage
\appendix
\section{Additional Notation}
A \emph{Markov decision process} is a special case of concurrent games
such that there exists a $p \in \set{\reach,\safe}$ such that for all
states~$\state \in \states$ holds $|\moveAssignment[p]|=1$ and a
\emph{Markov chain} is a special case of $\SG$ Markov decision
processes where for every $i \in \set{\reach,\safe}$ and for every
state~$\state$ holds $|\moveAssignment[i]|=1$.

\subsection{Generalised Notion of Expected Value}
We can compute the expected value for a given valuation~$\valuation$
and strategies $\strategy[\maximize]$ and $\strategy[\minimize]$ by

\begin{align}
  \preop[\strategy[\maximize],\strategy[\minimize]](\valuation)(\state)= &
  \sum_{\move[\maximize],\move[\minimize] \in \moves} \sum_{\state<\prime> \in
    \states} \valuation(\state<\prime>) \cdot
  \transitions(\state,\move[\maximize],\move[\minimize])(\state<\prime>) \cdot
  \strategy[\maximize](\state)(\move[\maximize]) \cdot
  \strategy[\minimize](\state)(\move[\minimize])
\end{align}

$\valuation(\state<\prime>)$ denotes the current estimate of the value
for state~$\state<\prime>$ and it is weighted by the probability to go
from state~$\state$ to state~$\state<\prime>$ given the
moves~$\move[\maximize]$ and~$\move[\minimize]$ and their probability
to be seen under strategies~$\strategy[\maximize]$ and
$\strategy[\minimize]$. This probability is computed by
$\transitions(\state,\move[\maximize],\move[\minimize])(\state<\prime>)
\cdot \strategy[\maximize](\state)(\move[\maximize]) \cdot
\strategy[\minimize](\state)(\move[\minimize])$. 
Please note that we deliberately avoid fixing players in the
definition of the pre-operator~$\preop$. Instead, we associate
$\maximize$ with the player, which tries to maximize the estimate, and
$\minimize$ with the player, which tries to minimize it. For
reachability, $\maximize$ will correspond to player~$\reach$
maximizing the reachability value, while for safety, it will
correspond to player~$\safe$ maximizing the safety value. Minimizing
and maximizing are added with supremum and infimum computations over
all strategies as follows:

\begin{align}
    \preop[\maximize:\strategy[\maximize]](\valuation)(\state)= & \inf_{\strategy[\minimize] \in \strategies[\minimize]} \preop[\strategy[\reach],\strategy[\safe]](\valuation)(\state) \\
    \preop[\maximize](\valuation)(\state)= & \sup_{\strategy[\maximize]\in\strategies[\maximize]} \preop[\maximize:\strategy[\maximize]](\valuation)(\state) \\
\end{align} 

Please note that $\strategies[\maximize]$ denotes the set of all
strategies for the maximizing player (and $\strategies[\minimize]$ the
set of strategies for the minimizing player). The computation of
$\preop[\maximize](\valuation)(\state)$ reduces to the solution of a
zero-sum one-shot matrix game and can be solved by linear
optimization. Optimal strategies in zero-sum one-shot games need
randomisation and we denote the strategy which can achieve the value
of $\preop[\maximize](\valuation)$ by
$\best[\maximize](\valuation)$.  We give an example of the payoff
matrix for a one-shot zero-sum game corresponding to
$\preop[\maximize](\valuation)(\state)$ in the next section.
 

\section{Correctness Proof for Value Iteration}
For this section, we fix a concurrent game $\exampleGame$. Moreover, we fix the following notation:
\[\best<\valuation>(\endComponent) \coloneqq \set{\state \in \endComponent \mid
  \preop[\reach]<\exit(\endComponent)>(\valuation)(\state) =
  \bestExit<\valuation>(\endComponent)}, \]
i.e. $\best<\valuation>(\endComponent)$ is the set of states in
$\endComponent$, which have an exiting strategy yielding the
value~$\bestExit<\valuation>(\endComponent)$.

In order to proof the correctnes of bounded value iteration, it
suffices to prove that the sequence $\upperBound<k>$ converges to the
actual value because the convergence of $\lowerBound<k>$ has already
been proven. That is also the reason why we can make use of the
following two claims:

\begin{itemize}
\item $\val(\Diamond \success)$ is a fixpoint of the operator
  $\preop[\reach]$, i.e. $\preop[\reach](\val(\Diamond \success)) =
  \val(\Diamond \success)$
\item The operators~$\preop$ and $\preop[\reach]$ are monotone, which
  we state in \Fref{lem:preopDecreases} for simpler usage.
\end{itemize}

We proceed in to steps. First, we prove that
$\upperBound<k>$ converges to a fixpoint. Afterwards, we will show
that this fixpoint coincides with the value.

\subsection{Convergence to a fixpoint $\upperBound<\star>$}
In order to prove the convergence to a fixpoint, we essentially need
to show that the sequence is bounded from below and monotonically
decreasing. First, we show that $\upperBound<k>$ is indeed an upper
bound of the value. The complex part of the proof is the correctness
of the \procname{DEFLATE}. In this procedure, we reduce the upper
bounds of some states to the best exit from a subset of states. We
start with following lemma, which states that decreasing the value
using \procname{DEFLATE} never decreases the current value below the
true value.

\begin{lemma}
  \label{lem:valBelowBestExit}
  For $\states<\prime> \subseteq \states \setminus (\success \cup
  \winning)$ and all states $\state \in \states<\prime>$ we have
  $\val(\Diamond \success)(\state) \leq \bestExit<\val(\Diamond
  \success)>(\states<\prime>)$.
\end{lemma}

\begin{proof}
  Let $\anySet := \set{\state \in \states<\prime> \mid \val(\Diamond
    \success)(\state) = \max_{\state<\prime> \in \states<\prime>}
    \val(\Diamond \success)(\state<\prime>)}$. Consider the following
  cases:
  \begin{description}
  \item[First Case:] For some $\state \in \anySet$ there is an optimal
    strategy $\strategy[\reach] \in
    \strategies[\reach]<\exit(\states<\prime>)>(\state)$. Then,
    $\val(\Diamond \success)(\state<\prime>) \leq \val(\Diamond
    \success)(\state) = \preop[\reach:\strategy[\reach]](\val(\Diamond
    \success))(\state) \leq \bestExit<\val(\Diamond
    \success)>(\states<\prime>)$ for all $\state<\prime> \in
    \states<\prime>$.

  \item[Second Case:] There exists no $\state \in \anySet$ with an
    optimal strategy $\strategy[\reach] \in q
    \strategies[\reach]<\exit(\states<\prime>)>(\state)$. We
    distinguish two cases:
    \begin{description}
    \item[Case A:] Let for all $\state \in \anySet$ and
      $\strategy[\reach]$ there is a $\strategy[\safe]$ such that
      $\destination(\state, \strategy[\reach], \strategy[\safe])
      \subseteq \states<\prime>$, then due to the maximality of the
      values for states in $\anySet$ we must also have
      $\destination(\state, \strategy[\reach], \strategy[\safe])
      \subseteq \anySet$. Since none of the states chooses a leaving
      strategy then and $\states<\prime> \cap \success = \emptyset$,
      it then holds $\val(\Diamond \success)(\state) = 0$ for all
      $\state \in \anySet$, which is a contradiction to the assumption
      (especially to $\states<\prime> \cap \winning = \emptyset$.

    \item[Case B:] Let there be a $\state \in \anySet$ with an optimal
      strategy $\strategy[\reach]$, such that for all
      $\strategy[\safe]$ we have $\destination(\state,
      \strategy[\reach], \strategy[\safe]) \not\subseteq
      \states<\prime>$. Then, there must be moves $\move[\reach] \in
      \support(\strategy[\reach](\state))$ s.t. for all $\move[\safe]$
      we have $\destination(\state, \move[\reach], \move[\safe])
      \subseteq \states<\prime>$ since we would have
      $\strategy[\reach] \in
      \strategies[\reach]<\exit(\states<\prime>)>(\state)$ otherwise,
      which is not the case by assumption. Let $\mathcal{M} :=
      \set{\move[\reach] \in \moves[\reach] \mid \forall \move[\safe]
        \in \moves[\safe].\ \destination(\state, \move[\reach],
        \move[\safe]) \subseteq \states<\prime>}$ be a set of moves
      for player~$\reach$ such that they lead to states
      in~$\states<\prime>$ for all moves of player~$\safe$. We
      modify~$\strategy$ in such a way that moves in~$\mathcal{M}$ are
      taken with probability~$0$:

      \[
      \strategy[\reach]<\prime>(\state)(\move[\reach]) := \begin{cases}
        0 & \text{ if } \move[\reach] \in \mathcal{M} \\
        \strategy[\reach](\state)(\move[\reach]) / (1-x) & \text{ otherwise}
      \end{cases}
      \]
      
      where $x \coloneqq\sum_{\move[\reach] \in \mathcal{M}}
      \strategy[\reach](\state)(\move[\reach])$. The
      strategy~$\strategy<\prime>$ is well defined as
      
      \begin{equation*}
        \begin{split}
          \sum_{\move[\reach] \in \moves[\reach]} \strategy[\reach]<\prime>(\state)(\move[\reach]) & = \sum_{\move[\reach] \in \moves[\reach] \setminus \mathcal{M}} \strategy[\reach](\state)(\move[\reach]) / (1-x) \\
          & = 1/(1-x) \cdot \sum_{\move[\reach] \in \moves[\reach] \setminus \mathcal{M}} \strategy[\reach](\state)(\move[\reach]) \\
          & = 1/(1-x) \cdot (1-x) \\
          & = 1
        \end{split}
      \end{equation*}

      Then we have $\destination(\state, \strategy[\reach]<\prime>,
      \strategy[\safe]) \not\subseteq \states<\prime>$ and especially,
      for all $\move[\reach] \in
      \support(\strategy[\reach]<\prime>(\state))$, it holds for every
      $\move[\safe] \in \support(\strategy[\safe](\state))$ that
      $\destination(\state, \move[\reach], \move[\safe]) \not\subseteq
      \states<\prime>$. By definition, it thus holds
      $\strategy[\reach]<\prime> \in
      \strategies[\reach]<\exit(\states<\prime>)>(\state)$. Since the
      value of states in $\anySet$ is maximal among the states in
      $\states<\prime>$, the moves in $\mathcal{M}$ can yield at most
      the value of the states in $\anySet$. Since the remaining moves
      in~$\support(\strategy<\prime>(\state))$ have been in
      $\support(\strategy(\state))$ as well and thus, must at least
      guarantee the same value (we could define a better strategy by
      choosing only moves in~$\mathcal{M}$ otherwise), we have
      $\preop[\reach:\strategy[\reach]<\prime>](\val(\Diamond
      \success))(\state) =
      \preop[\reach:\strategy[\reach]](\val(\Diamond
      \success))(\state)$. Therefore, for all states $\state<\prime>
      \in \states<\prime>$ we have $\val(\Diamond
      \success)(\state<\prime>) \leq \val(\Diamond \success)(\state) =
      \preop[\reach:\strategy[\reach]<\prime>](\val(\Diamond
      \success))(\state) \leq \bestExit<\val(\Diamond
      \success)>(\states<\prime>)$.
    \end{description}
  \end{description}
\end{proof}

We first prove that the sequence $\upperBound<k>$ is bounded from
below by the true value for states in end components. Later, we also
show that is bounded for states that are in no end component at all.

For a valuation $\valuation$, denote $\onestep{D}<i>(\valuation)$, and
$\onestep{D}<i>(\endComponent)$ for the sequences of $\valuation$,
resp. $\endComponent$ during \procname{DEFLATE}. Then, we denote
$\onestep{D}<|\endComponent|>(\valuation)$ for $\valuation$ at the end
of \procname{DEFLATE}.

The next lemma shows that during no iteration of \procname{DEFLATE} he
current value is deflated below the true value.

\begin{lemma}
    \label{lem:updateMECCorrectUpperBound}
    Assume $\upperBound \geq \val(\Diamond \success)$. For a maximal
    end component $\endComponent \subseteq \states$, all states
    $\state \in \endComponent$ and $i \geq 0$, we have
    $\onestep{D}<i>(\upperBound)(\state) \geq \val(\Diamond
    \success)(\state)$.
\end{lemma}

\begin{proof}
  We apply induction over $i$.
  \begin{description}
  \item[Induction Basis:] Let $i = 0$. Then,
    $\onestep{D}<0>(\upperBound)(\state) = \upperBound(\state) \geq
    \val(\Diamond \success)(\state)$ holds by assumption.

  \item[Induction Hypothesis:] For $i \geq 0$ and $\state
    \in \endComponent$ we have $\onestep{D}<i>(\upperBound)(\state)
    \geq \val(\Diamond \success)(\state)$.

  \item[Induction Step:] Consider $i + 1$.
    \begin{description}
    \item[Case A:] Let $\state \in \endComponent$. If $\state \not\in
      \onestep{D}<i>(\endComponent)$, then
      $\onestep{D}<i+1>(\upperBound)(\state) =
      \onestep{D}<i>(\upperBound)(\state) \overset{\text{I.H.}}{\geq}
      \val(\Diamond \success)(\state)$. The first equality holds
      because \procname{DEFLATE} does not affect states outside of
      $\onestep{D}<i>(\endComponent)$.
    \item[Case B:] Assume that $\state \in
      \onestep{D}<i>(\endComponent)$. Then,
      $\onestep{D}<i+1>(\upperBound)(\state) =
      \min(\onestep{D}<i>(\upperBound)(s),
      \bestExit<\onestep{D}<i>(\upperBound)>(\onestep{D}<i>(\endComponent)))
      \overset{\text{I.H.}}{\geq} \min(\val(\Diamond
      \success)(\state), \bestExit<\val(\Diamond
      \success)>(\onestep{D}<i>(\endComponent)))$. By
      \Fref{lem:valBelowBestExit} we have $\bestExit<\val(\Diamond
      \success)>(\onestep{D}<i>(\endComponent)) \geq \val(\Diamond
      \success)(\state)$. Therefore, we can conclude that
      $\onestep{D}<i+1>(\upperBound)(\state) \geq \val(\Diamond
      \success)(\state)$.
    \end{description}
  \end{description}
\end{proof}

Having proven \Fref{lem:updateMECCorrectUpperBound}, it is a matter of
a simple induction to show the overall boundedness.

\begin{lemma}
  \label{lem:upperBoundCorrectness}
  For all $k \geq 0$ we have $\upperBound<k>(\state) \geq
  \val(\Diamond \success)(\state)$ for all $\state \in \states
  \setminus (\success \cup \winning)$.
\end{lemma}

\begin{proof}
  We apply induction over $k$.
  \begin{description}
  \item[Induction Basis:] Let $k = 0$. We have $\upperBound<0>(\state)
    = 1 \geq \val(\Diamond \success)(\state)$ for all $\state \in
    \states \setminus (\success \cup \winning)$.
            
  \item[Induction Hypothesis:] Assume that $\upperBound<k>(\state)
    \geq \val(\Diamond \success)(\state)$ holds for all $\state \in
    \states \setminus (\success \cup \winning)$.
            
  \item[Induction Step:] If there is no end component $\endComponent$
    with $\state \in \endComponent$, then $\upperBound<k+1>(\state) =
    \preop[\reach](\upperBound<k>)(\state) \overset{\text{I.H.}}{\geq}
    \preop[\reach](\val(\Diamond \success))(\state) = \val(\Diamond
    \success)(\state)$. If $\state \in \endComponent$ for a maximal
    end component $\endComponent$, then $\upperBound<k+1>(\state) =
    \onestep{D}<|\endComponent|>(\preop[\reach](\upperBound<k>))(\state)$. With
    the induction hypothesis and monotonicty of $\preop[\reach]$ the
    claim follows from \Fref{lem:updateMECCorrectUpperBound}.
  \end{description}
\end{proof}

Now that we have shown that the upper bound is indeed a correct upper
bound of the value, we are ready to move our attention to
monotonicity. This is one of the most important statements for the
correctness and the proof is more involved than those of most of the
others, which is the reason why we have split it across several
lemmas. The next lemma intuitively shows that deflating is
order-preserving in the sense that when applied to two
valuations~$\upperBound[1], \upperBound[2]$ and $\upperBound[2] \leq
\upperBound[1]$, then also
$\onestep{D}<|\endComponent|>(\upperBound[2]) \leq
\onestep{D}<|\endComponent|>(\upperBound[1])$. However, we prove a
slightly more general version of monotonicity in
\Fref{lem:deflateMonotone}.

\begin{lemma}
  \label{lem:deflateMonotone}
  Let $\upperBound[1], \upperBound[2]$ be two valuations and
  $\endComponent$ an end component such that the following holds:
  
  \begin{itemize}
  \item $\preop[\reach](\upperBound[1]) \leq \upperBound[1]$ and
  \item $\upperBound[2] \leq
    \preop[\reach](\onestep{D}<|\endComponent|>(\upperBound[1]))$.
  \end{itemize}
  
  Then, $\onestep{D}<|\endComponent|>(\upperBound[2]) \leq
  \onestep{D}<|\endComponent|>(\upperBound[1])$.
\end{lemma}

\begin{proof}
  Evidently, \procname{DEFLATE} can only decrease the valuation so
  $\onestep{D}<|\endComponent|>(\upperBound[2]) \leq
  \upperBound[2]$. Hence, it suffices to show that $\upperBound[2]
  \leq \onestep{D}<|\endComponent|>(\upperBound[1])$. We show by
  induction over $i$ that $\upperBound[2] \leq
  \onestep{D}<i>(\upperBound[1])$ for all $i$.

  \begin{description}
  \item[Base Case:] Let $i=0$. It holds
    $\onestep{D}<0>(\upperBound[1]) = \upperBound[1] \geq
    \preop[\reach](\upperBound[1]) \geq
    \preop[\reach](\onestep{D}<|\endComponent|>(\upperBound[1])) \geq
    \upperBound[2]$.

  \item[Induction Hypothesis:] Assume that $\upperBound[2] \leq
    \onestep{D}<i>(\upperBound[1])$.

  \item[Induction Step:] Consider $i+1$.
    \begin{description}
    \item[Case A:] First, assume that $\state \not\in
      \onestep{D}<i>(\endComponent)$. In that case
      $\onestep{D}<i+1>(\upperBound[1])(\state) =
      \onestep{D}<i>(\upperBound[1])(\state) \overset{I.H.}{\geq}
      \upperBound[2](\state)$.
    \item[Case B:] Now, assume that $\state \in
      \onestep{D}<i>(\endComponent)$. Then,
      \[
      \onestep{D}<i+1>(\upperBound[1])(\state) =
      \min(\onestep{D}<i>(\upperBound[1])(\state),
      \bestExit<\onestep{D}<i>(\upperBound[1])>(\onestep{D}<i>(\endComponent))).
      \]
      \begin{description}
      \item[Case 1:] If $\onestep{D}<i+1>(\upperBound[1])(\state) =
        \onestep{D}<i>(\upperBound[1])(\state)$, then the claim again
        follows immediately from the induction hypothesis.
      \item[Case 2:] Otherwise, observe that for all $\state<\prime>
        \in \onestep{D}<i>(\endComponent)$ we have
        $\onestep{D}<i+1>(\upperBound[1])(\state<\prime>) \leq
        \bestExit<\onestep{D}<i>(\upperBound[1])>(\onestep{D}<i>(\endComponent))$
        due to the above (minimizing) update, but
        $\onestep{D}<i+1>(\upperBound[1])(\state) =
        \bestExit<\onestep{D}<i>(\upperBound[1])>(\onestep{D}<i>(\endComponent))$
        Therefore,
        \begin{equation*}
          \begin{split}
            \onestep{D}<i+1>(\upperBound[1])(\state) & = \bestExit<\onestep{D}<i>(\upperBound[1])>(\onestep{D}<i>(\endComponent)) \\
            & \overset{\ast}{\geq} \preop[\reach](\onestep{D}<i+1>(\upperBound[1]))(\state) \\
            & \geq \preop[\reach](\onestep{D}<|\endComponent|>(\upperBound[1]))(\state) \\
            & \geq \upperBound[2](\state).
          \end{split}
        \end{equation*}
        The step labeled with~$\ast$ holds due to our assumptions as
        well as the monotonicity of $\preop$.
      \end{description}
    \end{description}

  \end{description}
\end{proof}

\begin{lemma}
    \label{lem:preopDecreases}
    For $k \in \NN$, we have $\preop[\reach](\upperBound<k>) \leq \upperBound<k>$.
\end{lemma}

\begin{proof}
  In \cite{Chatterjee2013}.
\end{proof}

\begin{lemma}
  \label{lem:upperBoundMonotone}
  For $k \in \NN$, we have $\upperBound<k> \geq \upperBound<k+1>$. 
\end{lemma}

\begin{proof}

  \begin{description}
  \item[Induction Basis:] $k = 0$. By definition, we have
    $\upperBound<0>(\state) = \mathsf{if} \ \state \in \winning \
    \mathsf{then} \ 0 \ \mathsf{else} \ 1$, i.e. $\upperBound<0>$ is
    equal to~$1$ everywhere except $\winning$. Since we assume winning
    to be absorbing, $\preop[\reach](\upperBound)(\state) = 0$ for any
    state $\state \in \winning$. Since
    $\preop[\reach](\upperBound)(\state) \leq 1$, this also holds for
    any best exit. Hence, $\upperBound<0> \geq \upperBound<1>$.

  \item[Induction Hypothesis:] $\upperBound<k> \geq \upperBound<k+1>$.

  \item[Induction Step:] We need to prove that $\upperBound<k+1> \geq
    \upperBound<k+2>$.  Let $\state \in \states \setminus (\winning
    \cup \success)$. On $\winning$, $\upperBound<i>$ is equal to~$0$
    and on $\success$ to~$1$ for any $i \in \NN$.
    
    \begin{description}
    \item[Case A:] $\state$ is not contained in any (maximal) end
      component. Then, $\upperBound<k+2>(\state) =
      \preop[\reach](\upperBound<k+1>)(\state) \leq
      \upperBound<k+1>(\state)$ by \Fref{lem:preopDecreases}.

    \item[Case B:] $\state \in \endComponent$, for a maximal end
      component $\endComponent$. Denote $\upperBound[1]$ for
      $\upperBound<k+1>$ before calling \procname{DEFLATE}, and
      likewise let $\upperBound[2]$ be the valuation
      $\upperBound<k+2>$ before calling \procname{DEFLATE}. Formally,
      $\upperBound[1] = \preop[\reach](\upperBound<k>)$ and
      $\upperBound[2] = \preop[\reach](\upperBound<k+1>)$. By
      \Fref{lem:preopDecreases} we have
      $\preop[\reach](\upperBound<k>) \leq \upperBound<k>$. Therefore,
      \[\preop[\reach](\upperBound[1]) =
      \preop[\reach](\preop[\reach](\upperBound<k>)) \leq
      \preop[\reach](\upperBound<k>) = \upperBound[1]\]. Moreover,
      \[\upperBound[2] =
      \onestep{D}<|\endComponent|>(\preop[\reach](\upperBound<k+1>))
      \leq \preop[\reach](\upperBound<k+1>) =
      \preop[\reach](\onestep{D}<|\endComponent|>(\upperBound[1]))\]. Thus,
      the conditions of \Fref{lem:deflateMonotone} hold and we can
      conclude \[\upperBound<k+2>(\state) =
      \onestep{D}<|\endComponent|>(\preop[\reach](\upperBound<k+1>))(\state)
      = \onestep{D}<|\endComponent|>(\upperBound[2]))(\state) \leq
      \onestep{D}<|\endComponent|>(\upperBound[1])(\state) =
      \onestep{D}<|\endComponent|>(\preop[\reach](\upperBound<k>))(\state)
      = \upperBound<k+1>(\state)\].
    \end{description}
  \end{description}
\end{proof}

Finally, we are in the position to show the convergence to a
fixpoint. The proof is essentially the same as for Kleene's Fixpoint
Theorem. However, our valuations with the partial order $\leq$ are not
a lattice as not every set of valuations has a least element. Hence,
we can not simply apply the theorem. Instead, we show that we can
argue in a similar way for our setting.

\begin{theorem}
    \label{theo:upperBoundConverges}
    $\upperBound<\star> \coloneqq \lim_{k \to \infty} \upperBound<k>$
    exists and $F(\upperBound<\star>) = \upperBound<\star>$, where $F$
    denotes the update of $\upperBound$ in
    \Fref{alg:boundedValueIteration}.
\end{theorem}

Please note that $F(\upperBound<k>) = \upperBound<k+1>$.

\begin{proof}
  Let $\mathbb{U} := \set{\upperBound<k> \mid k \geq 0}$. By
  \Fref{lem:upperBoundCorrectness} it follows that $\mathbb{U}$ is
  bounded from below by $\val(\Diamond \success)$. Moreover, since $F$
  is montone by \Fref{lem:upperBoundMonotone}, there exists $\inf
  \mathbb{U} = \lim_{k \to \infty} \upperBound<k> =
  \upperBound<\star>$. Since for any $\upperBound \in \mathbb{U}$ we
  have $\inf \mathbb{U} \leq \upperBound$, due to the monotonicity of
  $F$, we also have $F(\inf \mathbb{U}) \leq F(\upperBound)$. Hence,
  $F(\inf \mathbb{U}) = \inf F(\mathbb{U})$. It follows that
  $F(\upperBound<\star>) = F(\inf \mathbb{U}) = \inf F(\mathbb{U}) =
  \inf \mathbb{U} = \upperBound<\star>$.
\end{proof}

In the subsequent proofs we will often implicitly make use of the fact that $\upperBound<\star>$ is a fixpoint.

\subsection{Uniqueness of the Fixpoint}
The aim of this section is to show that the fixpoint
$\upperBound<\star>$ coincides with the true value function. Doing so
in absence of end components is fairly straightforward. The
\procname{DEFLATE} procedure deals with the end components by reducing
the upper bound in a sound way, as we have proven in the previous
section. In order to prove that this fixpoint is equal to the value,
we need to establish some further properties about the
fixpoint. Intuitively, we expect the claim to hold because all changes
of the upper bound are propagated to all states---even those that are
in an end component. However, since our \procname{DEFLATE} procedure
might operate on subsets of states that are not end components, we
need to show certain properties for general subsets of the state
space. The lemmas in this section essentially state that the fixpoint
behaves as we expect it to and culminate in
\Fref{lem:exitingStateUpperBound}. The following lemma is quite
natural to expect.

\begin{lemma}
    \label{lem:fixpointIsPreop}
    For all $\state \in \states$ we have $\upperBound<\star>(\state) =
    \preop[\reach](\upperBound<\star>)(\state)$.
\end{lemma}

\begin{proof}
  For all states $\state \in \states \setminus (\winning \cup \success)$ we compute
  $\upperBound<k+1>(\state) \coloneqq
  \preop[\reach](\upperBound<k>)(\state)$. Depending on whether or
  not $\state$ is in a maximal end component, we do or do not process
  it further. Consider the cases:
    \begin{description}
        \item[Case A:] There is no maximal end component~$\endComponent$ s.t. $\state \in \endComponent$. Then, $\upperBound<k+1>(\state) = \preop[\reach](\upperBound<k>)(\state)$. Since $\upperBound<\star>$ is a fixpoint~(\Fref{theo:upperBoundConverges}), it must hold that $\upperBound<\star>(\state) = \preop[\reach](\upperBound<\star>)(\state)$.

        \item[Case B:] Let $\endComponent \subseteq \states$ be a maximal end component with $\state \in \endComponent$. Then, we have $\upperBound<k+1>(\state) = \onestep{D}<|\endComponent|>(\preop[\reach](\upperBound<k>))(\state)$. By \Fref{lem:upperBoundMonotone} we have $\upperBound<k>(\state) \geq \onestep{D}<|\endComponent|>(\preop[\reach](\upperBound<k>))(\state) \geq \preop[\reach](\upperBound<k>)(\state)$. Since $\upperBound<\star>$ is the limit of this sequence, we have $\upperBound<\star>(\state) \geq \preop[\reach](\upperBound<\star>)(\state)$. On the other hand, since $\upperBound<\star>$ is a fixpoint, we have $\upperBound<\star>(\state) = \onestep{D}<|\endComponent|>(\preop[\reach](\upperBound<\star>))(\state) \leq \preop[\reach](\upperBound<\star>)(\state)$. Combining the two, we obtain $\upperBound<\star>(\state) = \preop[\reach](\upperBound<\star>)(\state)$, which was to prove.
    \end{description}
\end{proof}

During \procname{DEFLATE} we reduce the upper bound of states in an end component in layers to the best exit from the current subset of the end component. In the previous section, we have already justified why this is sound to do, in the sense that we would never decrease the upper bound below the value. With the following lemma we show that in the limit we can really attain the value of the best exit.

\begin{lemma}
    \label{lem:fixpointInECIsBestExit}
    Let $\endComponent \subseteq \states$ be an end component, and $\onestep{D}<i>(\endComponent )\subseteq \endComponent$. Then, for all $\state \in \attractor(\best<\upperBound<\star>>(\onestep{D}<i>(\endComponent)))$ we have $\upperBound<\star>(\state) = \bestExit<\upperBound<\star>>(\onestep{D}<i>(\endComponent))$.
\end{lemma}

\begin{proof}
    Since $\upperBound<\star>$ is a fixpoint and by definition of \procname{DEFLATE} we have 
    \[
        \upperBound<\star>(\state) = \min(\upperBound<\star>(\state), \bestExit<\upperBound<\star>>(\onestep{D}<i>(\endComponent))) \leq \bestExit<\upperBound<\star>>(\onestep{D}<i>(\endComponent))
    \]
    for all $\state \in \onestep{D}<i>(\endComponent)$. We show that for all $\state \in \attractor<k>(\best<\upperBound<\star>>(\onestep{D}<i>(\endComponent)))$ we have $\upperBound<\star>(\state) \geq \bestExit<\upperBound<\star>>(\onestep{D}<i>(\endComponent))$ by induction over $k$.
    \begin{description}
        \item[Induction Basis:] Let $k = 0$ and $\state \in \attractor<0>(\best<\upperBound<\star>>(\onestep{D}<i>(\endComponent))) = \best<\upperBound<\star>>(\onestep{D}<i>(\endComponent))$. Then, we have $\preop[\reach]<\exit(\onestep{D}<i>(\endComponent))>(\upperBound<\star>)(\state) = \bestExit<\upperBound<\star>>(\onestep{D}<i>(\endComponent))$. Moreover $\upperBound<\star>(\state) \overset{\ast}{=} \preop[\reach](\upperBound<\star>)(\state) \geq \preop[\reach]<\exit(\onestep{D}<i>(\endComponent))>(\upperBound<\star>)(\state)$, where $\ast$ follows from \Fref{lem:fixpointIsPreop}. It follows that $\upperBound<\star>(\state) \geq \bestExit<\upperBound<\star>>(\onestep{D}<i>(\endComponent))$.

        \item[Induction Hypothesis:] For all $\state \in \attractor<k>(\best<\upperBound<\star>>(\onestep{D}<i>(\endComponent)))$ we have $\upperBound<\star>(\state) = \bestExit<\upperBound<\star>>(\onestep{D}<i>(\endComponent))$.

        \item[Induction Step:] Let $\state \in \attractor<k+1>(\best<\upperBound<\star>>(\onestep{D}<i>(\endComponent)))$. By definition of $\attractor$ there is a $\strategy[\reach] \in \strategies[\reach]$, s.t. for all $\strategy[\safe] \in \strategies[\safe]$ we have $\destination(\state, \strategy[\reach], \strategy[\safe]) \subseteq \attractor<k>(\best<\upperBound<\star>>(\onestep{D}<i>(\endComponent)))$. From the induction hypothesis follows that for all $\state<\prime> \in \attractor<k>(\best<\upperBound<\star>>(\onestep{D}<i>(\endComponent)))$ we have $\upperBound<\star>(\state<\prime>) \geq \bestExit<\upperBound<\star>>(\onestep{D}<i>(\endComponent))$. Hence, $\upperBound<\star>(\state) \overset{\ast}{=} \preop[\reach](\upperBound<\star>)(\state) \geq \preop[\reach:\strategy[\reach]](\upperBound<\star>)(\state) \geq \bestExit<\upperBound<\star>>(\onestep{D}<i>(\endComponent))$, where $\ast$ again holds by \Fref{lem:fixpointIsPreop}.
    \end{description}
\end{proof}

The following two lemmas are essentially some technical overhead required to prove \Fref{lem:exitingStateUpperBound}, which can be considered the main step in proving the correctness of the fixpoint. All complications stemming from end components arise from overestimating the value within the end component and therefore not taking into account what happens outside of it. With \Fref{lem:exitingStateUpperBound} we show that the fixpoint of our algorithm always takes into account the values of states outside any subset of the states.

\begin{lemma}
    \label{lem:attractorBestExitInSubset}
    Let $\states<\prime> \subseteq \onestep{D}<i>(\endComponent) \subseteq \endComponent$, where $\endComponent$ is a maximal end component, and $\states<\prime> \cap \attractor(\best<\upperBound<\star>>(\onestep{D}<i>(\endComponent))) \not= \emptyset$. Then, there exists a state \mbox{$\state \in \states<\prime> \cap \attractor(\best<\upperBound<\star>>(\onestep{D}<i>(\endComponent)))$}, s.t. there is a $\strategy[\reach] \in \strategies[\reach]<\exit( \states<\prime>)>(\state)$ with $\preop[\reach:\strategy[\reach]](\upperBound<\star>)(\state) = \bestExit<\upperBound<\star>>(\onestep{D}<i>(\endComponent))$.
\end{lemma}

\begin{proof}
    By definition $\attractor = \attractor<|\onestep{D}<i>(\endComponent)|>$. We apply induction over $k$ to show that whenever $\states<\prime> \cap \attractor<k>(\best<\upperBound<\star>>(\onestep{D}<i>(\endComponent))) \not= \emptyset$, then there exists a state $\state \in \states<\prime> \cap \attractor<k>(\onestep{D}<i>(\endComponent))$ with the desired property.
    \begin{description}
        \item[Induction Basis:] Let $k = 0$. Then, there is a $\state \in \states<\prime> \cap \best<\upperBound<\star>>(\onestep{D}<i>(\endComponent))$. Hence, there must be a $\strategy[\reach] \in \strategies[\reach]<\exit(\onestep{D}<i>(\endComponent))>(\state)$ with $\preop[\reach:\strategy[\reach]](\upperBound<\star>)(\state) = \bestExit<\upperBound<\star>>(\onestep{D}<i>(\endComponent))$. Since $\states<\prime> \subseteq \onestep{D}<i>(\endComponent)$, it is also the case that $\strategy[\reach] \in \strategies[\reach]<\exit(\states<\prime>)>(\state)$, which was to construct.

        \item[Induction Hypothesis:] If $\states<\prime> \cap \attractor<k>(\best<\upperBound<\star>>(\onestep{D}<i>(\endComponent))) \not= \emptyset$, there must be a state $\state \in \states<\prime> \cap \attractor<k>(\best<\upperBound<\star>>(\onestep{D}<i>(\endComponent)))$ s.t. there is a strategy $\strategy[\reach] \in \strategies[\reach]<\exit(\states<\prime>)>(\state)$ with $\preop[\reach:\strategy[\reach]](\upperBound<\star>)(\state) = \bestExit<\upperBound<\star>>(\onestep{D}<i>(\endComponent))$.

        \item[Induction Step:] Consider $k + 1$. Let $\state \in \states<\prime> \cap \attractor<k+1>(\best<\upperBound<\star>>(\onestep{D}<i>(\endComponent)))$. Then, by definition of $\attractor$ there is a strategy $\strategy[\reach] \in \strategies[\reach]$, s.t. for all $\strategy[\safe] \in \strategies[\safe]$ we have $\destination(\state, \strategy[\reach], \strategy[\safe]) \subseteq \attractor<k>(\best<\upperBound<\star>>(\onestep{D}<i>(\endComponent)))$. By \Fref{lem:fixpointInECIsBestExit} we have $\upperBound<\star>(\state<\prime>) = \bestExit<\upperBound<\star>>(\onestep{D}<i>(\endComponent))$ for all $\state<\prime> \in \attractor(\best<\upperBound<\star>>(\onestep{D}<i>(\endComponent))) \supseteq \attractor<k>(\best<\upperBound<\star>>(\onestep{D}<i>(\endComponent)))$. Hence, $\preop[\reach:\strategy[\reach]](\upperBound<\star>)(\state) = \bestExit<\upperBound<\star>>(\onestep{D}<i>(\endComponent))$. If $\strategy[\reach] \in \strategies[\reach]<\exit(\states<\prime>)>(\state)$, then we are done. Otherwise, $\states<\prime> \cap \attractor<k>(\best<\upperBound<\star>>(\onestep{D}<i>(\endComponent))) \not= \emptyset$ and the claim follows from the induction hypothesis.
    \end{description}
\end{proof}

\begin{lemma}
  \label{lem:exitSubsetYieldsLargerValue}
  Let $\states<\prime> \subseteq \endComponent$. For every state
  $\state \in \states<\prime>$ and every valuation $\valuation$, we
  have
  \begin{center}
    $\preop[\reach]<\exit(\states<\prime>)>(\valuation)(\state)
    \geq
    \preop[\reach]<\exit(\endComponent)>(\valuation)(\state).$
  \end{center}
\end{lemma}

\begin{proof}
  \begin{description}
  \item We have $\strategies[\reach]<\exit(\states<\prime>)>(\state)
    \supseteq \strategies[\reach]<\exit(\endComponent)>(\state)$,
    since every strategy leaving $\endComponent$ must also leave
    $\states<\prime>$ ($\states<\prime> \subseteq \endComponent$), but
    a strategy leaving~$\states<\prime>$ might still lead to states,
    which are all in $\endComponent$ and thus, may not leave
    $\endComponent$. With the fact that for sets $A \supseteq B$ holds
    $\sup A \geq \sup B$, the claim follows immediately from the
    definition of $\preop[\reach]$.
  \end{description}
\end{proof}

\begin{lemma}
    \label{lem:exitingStateUpperBound}
    For all $\states<\prime> \subseteq \states$ there is a $\state \in \states<\prime>$ with $\upperBound<\star>(\state) = \preop[\reach]<\exit(
    \states<\prime>)>(\upperBound<\star>)(\state)$.
\end{lemma}

\begin{proof}
    By \Fref{lem:fixpointIsPreop} we have $\upperBound<\star>(\state) = \preop[\reach](\upperBound<\star>)(\state) \geq \preop[\reach]<\exit(\states<\prime>)>(\upperBound<\star>)(\state)$. Therefore, it suffices to prove that $\upperBound<\star>(\state) \leq \preop[\reach]<\exit(\states<\prime>)>(\upperBound<\star>)(\state)$. We distinguish the following cases:
    \begin{description}
        \item[First Case:] $\states<\prime>$ is not an end component.
            \begin{description}
              \item Then, for all $\strategy[\reach] \in \strategies[\reach]$ and $\strategy[\safe] \in \strategies[\safe]$, there is $\state \in \states<\prime>$, such that we have \mbox{$\destination(\state, \strategy[\reach], \strategy[\safe]) \not\subseteq \states<\prime>$}. Note we can equivalently say that there exists a state $\state \in \states<\prime>$, such that for all strategies $\strategy[\reach] \in \strategies[\reach]$ and $\strategy[\safe] \in \strategies[\safe]$, we have $\destination(\state, \strategy[\reach], \strategy[\safe]) \subseteq \states<\prime>$. \footnote{Of course, this is not a general logical equivalence, but in this case it is not difficult to see that it holds.}
              \item[Case A:] There exists no (maximal) end component $\endComponent$ with $\state \in \endComponent$.  Then $\upperBound<\star>(\state) = \preop[\reach](\upperBound<\star>)(\state) = \preop[\reach]<\exit( \states<\prime>)>(\upperBound<\star>)$. The first equality holds by the definition of \procname{BVI} and from the fact that $\upperBound<\star>$ is a fixpoint, and the second equality holds because every strategy is exiting.
              \item[Case B:] Assume that $\state \in \endComponent$, where $\endComponent$ is a maximal end component (i.e. there exists a maximal end component, which shares some states with $\states<\prime>$).
                  \begin{description}
                      \item Let $\onestep{D}<i>(\endComponent )\subseteq \endComponent$, s.t. $\state \in \attractor(\best<\upperBound<\star>>(\onestep{D}<i>(\endComponent)))$.
                      \item[Case B1:] If $\state \in \best<\upperBound<\star>>(\onestep{D}<i>(\endComponent))$, then $\valuation \coloneqq \preop[\reach]<\exit(\onestep{D}<i>(\endComponent))>(\upperBound<\star>)(\state) = \bestExit<\upperBound<\star>>(\onestep{D}<i>(\endComponent))$.
                      \item[Case B2:] Otherwise, there is a strategy $\strategy[\reach]$, s.t. for all $\strategy[\safe]$ we have $\destination(\state, \strategy[\reach], \strategy[\safe]) \subseteq \attractor(\best<\upperBound<\star>>(\onestep{D}<i>(\endComponent)))$ by the definition of the attractor.  By \Fref{lem:fixpointInECIsBestExit} we have $\upperBound<\star>(\state<\prime>) = \bestExit<\upperBound<\star>>(\onestep{D}<i>(\endComponent))$ for all $\state<\prime> \in \attractor(\best<\upperBound<\star>>(\onestep{D}<i>(\endComponent)))$.  Hence, $\valuation \coloneqq \preop[\reach:\strategy[\reach]](\upperBound<\star>)(\state) = \bestExit<\upperBound<\star>>(\onestep{D}<i>(\endComponent))$.  
                      \item In both cases, we have \mbox{$\preop[\reach]<\exit( \states<\prime>)>(\upperBound<\star>)(\state) = \preop[\reach](\upperBound<\star>)(\state) \geq \valuation = \bestExit<\upperBound<\star>>(\onestep{D}<i>(\endComponent)) = \upperBound<\star>(\state)$}
                  \end{description}
            \end{description}
        \item[Second Case:] $\states<\prime>$ is an end component. Then, let $\states<\prime> \subseteq \endComponent$, where $\endComponent$ is a maximal end component. Moreover, let $\onestep{D}<i>(\endComponent )\subseteq \endComponent$ be the first subset during \procname{DEFLATE} such that $\state \in \attractor(\best<\upperBound<\star>>(\onestep{D}<i>(\endComponent)))$ holds for some $\state \in \states<\prime>$.  Then, we have that $\states<\prime> \subseteq \onestep{D}<i>(\endComponent)$ and $\upperBound<\star>(\state) = \bestExit<\upperBound<\star>>(\onestep{D}<i>(\endComponent))$ holds by \Fref{lem:fixpointInECIsBestExit}.
            \begin{description}
                \item[Case A:] If $\state \in \best<\upperBound<\star>>(\onestep{D}<i>(\endComponent))$, then $\preop[\reach]<\exit(\onestep{D}<i>(\endComponent))>(\upperBound<\star>)(\state) = \bestExit<\upperBound<\star>>(\onestep{D}<i>(\endComponent))$. 
                \item By \Fref{lem:exitSubsetYieldsLargerValue} we then have $\preop[\reach]<\exit( \states<\prime>)>(\upperBound<\star>)(\state) \geq \upperBound<\star>(\state)$. 
                \item[Case B] Otherwise, by assumption we have $\states<\prime> \cap \attractor(\best<\upperBound<\star>>(\onestep{D}<i>(\endComponent))) \not= \emptyset$. Therefore, \Fref{lem:attractorBestExitInSubset} yields a state $\state<\prime> \in \states<\prime> \cap \attractor(\best<\upperBound<\star>>(\onestep{D}<i>(\endComponent)))$ and a strategy $\strategy[\reach] \in \strategies[\reach]<\exit( \states<\prime>)>(\state<\prime>)$, with $\preop[\reach:\strategy[\reach]](\upperBound<\star>)(\state<\prime>) = \bestExit<\upperBound<\star>>(\onestep{D}<i>(\endComponent))$. We conclude $\preop[\reach]<\exit( \states<\prime>)>(\upperBound<\star>)(\state<\prime>) \geq \preop[\reach:\strategy[\reach]](\upperBound<\star>)(\state<\prime>) = \bestExit<\upperBound<\star>>(\onestep{D}<i>(\endComponent)) = \upperBound<\star>(\state<\prime>)$. The last equality follows from the fact that $\state<\prime> \in \attractor(\best(\onestep{D}<i>(\endComponent)))$ and \Fref{lem:fixpointInECIsBestExit}.
          \end{description}
    \end{description}
\end{proof}

Having proven \Fref{lem:exitingStateUpperBound} it is not difficult to establish the main result of this section:

\begin{theorem}
    \label{theo:upperBoundConvergesToValue}
    For $\state \in \states$ have $\upperBound<\star>(\state) = \val(\Diamond \success) (\state)$.
\end{theorem}
\begin{proof}
  Assume there is a state $\state \in \states$, s.t. $\differ(\state) \coloneqq \upperBound<\star>(\state) - \val(\Diamond \success) (\state) > 0$. Let $\differ[max] \coloneqq \max_{\state \in \states} \differ(\state)$ and $\states<\prime> \coloneqq \set{\state \in \states \mid \differ(\state) = \differ[max]}$. We have $\success \cap \states<\prime> = \emptyset$ and $\winning \cap \states<\prime> = \emptyset$ since the estimates of the reachability probabilities for $\success$ and $\winning$ are correct throughout all iterations (and thus, also in the limit).

    From \Fref{lem:exitingStateUpperBound} we obtain a $\state \in \states<\prime>$ and $\strategy[\reach] \in \strategies[\reach]<\exit(\states<\prime>)>(\state)$ with $\upperBound<\star>(\state) = \preop[\reach:\strategy[\reach]](\upperBound<\star>)(\state)$. Note that for all $\strategy[\safe] \in \strategies[\safe]$ we have $\preop[\strategy[\reach], \strategy[\safe]](\differ)(\state) < \differ(\state)$ because for all states $\state<\prime> \in \destination(\state, \strategy[\reach], \strategy[\safe])$ we have $\differ(\state<\prime>) \leq \differ(\state)$ and for at least one $\state<\prime> \in \destination(\state, \strategy[\reach], \strategy[\safe])$ we have $\differ(\state<\prime>) < \differ(\state)$. The former follows from the fact that $\differ(\state) = \differ[max]$, and the latter follows from the fact that $\strategy[\reach]$ is an exiting strategy. Hence we have

    \begin{equation*}
        \begin{split}
            \upperBound<\star>(\state) & = \preop[\reach:\strategy[\reach]](\upperBound<\star>)(\state) \\
            & = \inf_{\strategy[\safe] \in \strategies[\safe]} \sum_{\move[\reach], \move[\safe] \in \moves} \sum_{\state<\prime> \in \states} \upperBound<\star>(\state<\prime>) \cdot \strategy[\reach](\move[\reach]) \cdot \strategy[\safe](\move[\safe]) \cdot \transitions(\state, \move[\reach], \move[\safe])(\state<\prime>) \\
            & = \inf_{\strategy[\safe] \in \strategies[\safe]} \sum_{\move[\reach], \move[\safe] \in \moves} \sum_{\state<\prime> \in \states} (\val(\Diamond \success) (\state<\prime>) + \differ(\state<\prime>)) \cdot \strategy[\reach](\move[\reach]) \cdot \strategy[\safe](\move[\safe]) \cdot \transitions(\state, \move[\reach], \move[\safe])(\state<\prime>) \\
            & = \inf_{\strategy[\safe] \in \strategies[\safe]} \preop[\strategy[\reach], \strategy[\safe]](\val(\Diamond \success))(\state) + \preop[\strategy[\reach], \strategy[\safe]](\differ)(\state) \\
            & < \inf_{\strategy[\safe] \in \strategies[\safe]} \preop[\strategy[\reach], \strategy[\safe]](\val(\Diamond \success))(\state) + \differ(\state) \\
            & \leq \sup_{\strategy[\reach] \in \strategies[\reach]} \inf_{\strategy[\safe] \in \strategies[\safe]} \preop[\strategy[\reach], \strategy[\safe]](\val(\Diamond \success))(\state) + \differ(\state) \\
            & = \val(\Diamond \success) (\state) + \differ(\state) \\
            & = \upperBound<\star>(\state)
        \end{split}
    \end{equation*}

    We obtain the inequality $\upperBound<\star>(\state) < \upperBound<\star>(\state)$ which is a contradiction.
\end{proof}

\section{Correctness Proof for Strategy Iteration}
Let $\exampleGame$ be a concurrent game with reachability
objective~$\success$ and safety objective~$\fail$ s.t. $\fail
\disjointUnion \success = \states$. W.l.o.g., we assume that
$\winning$, i.e. the winning region of player~$\safe$, and $\success$
are absorbing. 

\begin{lemma}
  \label{lem:indedual}
  Let $\valuation$ be a valuation on
  $\states$. $\procname{DEFLATE}(1-\valuation,\endComponent) =
  1-\procname{INFLATE}(\valuation,\endComponent)$\footnote{In the following, we
    assume that \procname{DEFLATE} and \procname{INFLATE} return the
    updated valuation ($\upperBound<k+1>$ in the case of
    \procname{DEFLATE} and $\valuation[i]$ in the case of
    \procname{INFLATE}). We also ignore the updates to
    $\strategy[\safe]<\ast>$ in \procname{INFLATE} since they do not
    change the valuation.}.
\end{lemma}

\begin{proof}
  For all states $\state \not\in\endComponent$ the claim follows since
  $\valuation(\state)$ remains unchanged (in both algorithms).

  In the following, we only consider $\state \in \endComponent$. For
  $\procname{DEFLATE}(1-\valuation,\endComponent)$ and
  $\procname{INFLATE}(\valuation,\endComponent)$, the computation of
  $B$ in every iteration conincides since the computation of the
  set~$B$ is identical. Hence, the claim is proven by
  \[\begin{array}{rl} 

   1-\procname{INFLATE}(\valuation<i>,\endComponent)(\state) 
   & = 1-\max(\valuation<i>(\state),1-\bestExit[\reach]<1-\valuation<i>>(\endComponent<i>))\\
   & = \min(1-\valuation<i>(\state)),1-(1-\bestExit[\reach]<1-\valuation<i>>(\endComponent<i>))) \\
   & = \min(1-\valuation<i>(\state),\bestExit[\reach]<1-\valuation<i>>(\endComponent<i>)) \\
   & = \procname{DEFLATE}(1-\valuation<i>,\endComponent)
 \end{array}\] for all states
  $\state \in \endComponent$, where $i$ denotes the iteration, in
  which $\state$ gets updated (both in \procname{DEFLATE} and
  \procname{INFLATE} since the computation of $B$ always coincides).
\end{proof}

\begin{corollary}
  Let $\valuation \leq \valuation<\prime>$. Then,
  $\procname{INFLATE}(\valuation,\endComponent) \leq
  \procname{INFLATE}(\valuation<\prime>,\endComponent)$.
\end{corollary}

\begin{proof}
  For $\valuation \leq \valuation<\prime>$ implies $1-\valuation \geq
  1-\valuation<\prime>$. Hence, it holds\footnote{We
    have proven that \procname{DEFLATE} is monotone in
    \Fref{lem:upperBoundMonotone} since we did not use anything except
    $\upperBound<k> \geq \upperBound<k+1>$ within the proof.}
  \[\procname{DEFLATE}(1-\valuation,\endComponent) \geq
  \procname{DEFLATE}(1-\valuation<\prime>,\endComponent).\] By
  \Fref{lem:indedual}, it holds
  $1-\procname{INFLATE}(\valuation,\endComponent) \geq
  1-\procname{INFLATE}(\valuation<\prime>,\endComponent)$ and thus,
  $\procname{INFLATE}(\valuation,\endComponent) \leq
  \procname{INFLATE}(\valuation<\prime>,\endComponent)$.
\end{proof}

\begin{lemma}
  \label{lem:preopdual}
  Let $\valuation$ be a valuation. Then,
  $1-\preop[\maximize,\reach](\valuation) = \preop[\maximize,\safe](1-\valuation)$. 
\end{lemma}

\begin{proof}
  Let $\state \in \states$. In the following, we use that for the
  one-shot matrix game $\preop[\maximize](\valuation)(\state)$, we can
  swap $\sup$ and $\inf$.

  \begin{align*} 
    & 1-\preop[\maximize,\reach](\valuation)(\state) \\ 
    = & 1-\big(\adjustlimits\sup_{\strategy[\reach]}\inf_{\strategy[\safe]}
    \sum_{\move[\reach],\move[\safe] \in \moves}
    \sum_{\state<\prime> \in \states} \valuation(\state<\prime>) \cdot
    \transitions(\state,\move[\reach],\move[\safe])(\state<\prime>) \cdot
    \strategy[\reach](\state)(\move[\reach]) \cdot
    \strategy[\safe](\state)(\move[\safe] )\big) \\
      \overset{\ast}{=} & \big(\adjustlimits\inf_{\strategy[\reach]}\sup_{\strategy[\safe]}
    \sum_{\move[\reach],\move[\safe] \in \moves}
    \sum_{\state<\prime> \in \states} (1 - \valuation(\state<\prime>)) \cdot
    \transitions(\state,\move[\reach],\move[\safe])(\state<\prime>) \cdot
    \strategy[\reach](\state)(\move[\reach]) \cdot
    \strategy[\safe](\state)(\move[\safe] )\big) \\
    = &
    \big(\adjustlimits\sup_{\strategy[\safe]}\inf_{\strategy[\reach]}
    \sum_{\move[\reach],\move[\safe] \in \moves}
    \sum_{\state<\prime> \in \states} (1 - \valuation(\state<\prime>)) \cdot
    \transitions(\state,\move[\reach],\move[\safe])(\state<\prime>) \cdot
    \strategy[\reach](\state)(\move[\reach]) \cdot
    \strategy[\safe](\state)(\move[\safe] )\big) \\
    = & \preop[\maximize,\safe](1-\valuation)(\state)\\
  \end{align*}

    For $\ast$ note that

    \begin{multline*}
        1 - \sum_{\move[\reach],\move[\safe] \in \moves}
        \sum_{\state<\prime> \in \states} \valuation(\state<\prime>) \cdot
        \transitions(\state,\move[\reach],\move[\safe])(\state<\prime>) \cdot
        \strategy[\reach](\state)(\move[\reach]) \cdot
        \strategy[\safe](\state)(\move[\safe]) = \\
        \sum_{\move[\reach],\move[\safe] \in \moves}
        \sum_{\state<\prime> \in \states} (1 - \valuation(\state<\prime>)) \cdot
        \transitions(\state,\move[\reach],\move[\safe])(\state<\prime>) \cdot
        \strategy[\reach](\state)(\move[\reach]) \cdot
        \strategy[\safe](\state)(\move[\safe]) 
    \end{multline*}

    because

    \[
        \sum_{\move[\reach],\move[\safe] \in \moves}
        \sum_{\state<\prime> \in \states}
        \transitions(\state,\move[\reach],\move[\safe])(\state<\prime>) \cdot
        \strategy[\reach](\state)(\move[\reach]) \cdot
        \strategy[\safe](\state)(\move[\safe]).
    \]

\end{proof}

\begin{proof}[Proof of \Fref{theo:siOptimal}]
  We define $\valuation*<k>(\state) = 1-\upperBound<k>(\state)$. By
  \Fref{lem:preopdual}, we know that $\valuation*<k+1>(\state)
  \coloneqq \preop[\maximize,\safe](\valuation*<k>)(\state) =
  1-\preop[\maximize,\reach](\upperBound<k>)(\state)$. In addition, we
  know that $\procname{INFLATE}(\valuation*<k>,\endComponent) =
  1-\procname{DEFLATE}(\upperBound<k>,\endComponent)$ for any maximal
  end component~$\endComponent \in \mathcal{M}$ by
  \Fref{lem:indedual}. Thus, $\displaystyle{\lim_{k \to \infty}} \valuation*<k> =
  \val(\Box \fail)$ by determination of concurrent reachability and
  safety games.

  In \procname{INFLATE}, we modify $\strategy[\safe]<k>$ to be the
  best response to player~$\reach$ enforcing to leave the end
  component, i.e. $\strategy[\safe]<k>$ enforces $\valuation<k>$ for
  player~$\safe$.

  Since $\valuation*<0> \leq \valuation<0>$ for any
  $\strategy[\safe]<0>$ and both $\preop$ and $\procname{INFLATE}$ are
  monotone, we can inductively prove that $\valuation*<k> \leq
  \valuation<k>$. By defintion $\valuation<k> \leq \val(\Box
  \fail)$. Hence, we have $\valuation*<k> \leq \valuation<k> \leq
  \val(\Box \fail)$ and $\val(\Box \fail) \leq \valuation<\infty> \leq \val(\Box
  \fail)$ in the limit. Hence, $\lim_{k\to\infty} \valuation<k> =
  \val(\Box \fail)$ and $\valuation<k>$ approximates the safety value
  monotonically from below.
\end{proof}

\end{document}